\documentclass[12pt]{article}

\usepackage{amsmath}
\usepackage{amssymb}
\usepackage{amsfonts}
\usepackage{dsfont} 
\usepackage{tikz}
\usepackage{subcaption}

\usepackage{graphicx}

\usepackage{amsmath,amssymb,amsthm,mathtools}
\usepackage{bm}
\usepackage{enumitem}
\usepackage{hyperref}
\usepackage{ulem}

\newtheorem{theorem}{Theorem}
\newtheorem{proposition}[theorem]{Proposition}
\newtheorem{corollary}[theorem]{Corollary}
\newtheorem{remark}[theorem]{Remark}

\usepackage{xurl}
\usepackage{hyperref}

\newcommand{\p}{\partial}

\newcommand{\be}{\begin{equation}}
\newcommand{\bel}[1]{\begin{equation}\label{#1}}
\newcommand{\ee}{\end{equation}}
\newcommand{\bea}{\begin{eqnarray}}
\newcommand{\eea}{\end{eqnarray}}
\newcommand{\balign}{\begin{align}}
\newcommand{\ealign}{\end{align}}
\newcommand{\ba}{\begin{array}}
\newcommand{\ea}{\end{array}}
\newcommand{\bfig}{\begin{figure}}
\newcommand{\efig}{\end{figure}}

\newcommand{\eref}[1]{(\ref{#1})}

\newcommand{\exval}[1]{\mbox{$\langle \, {#1}\, \rangle$}}

\newcommand{\rme}{\mathrm{e}}

\newcommand{\R}{{\mathbb R}}

\newcommand{\Z}{{\mathbb Z}}

\newcommand{\T}{{\mathbb T}} 

\newcommand{\balpha}{\boldsymbol{\alpha}}

\newcommand{\brho}{\boldsymbol{\rho}}

\newtheorem{theo}{Theorem}[section]

\newcommand{\Proof}{\noindent {\it Proof: }}
\def\qed{\hfill$\Box$\par\medskip\par\relax}

\begin{document}

\title{Hydrodynamics and boundary-induced phase transitions in the $n$-species
particle-exchange process}

\author{G. M. Sch\"utz$^{1}$ \and A. Zahra$^{1,2}$\thanks{\href{mailto:ali.zahra@univ-lorraine.fr}{ali.zahra@univ-lorraine.fr}}}

\vspace{10pt}

\vspace{10pt}
\maketitle

\noindent $^{1}$ Centro de An\'alise Matem\'atica, Geometria e Sistemas Din\^amicos, Departamento de Matem\'atica, Instituto Superior T'ecnico, Universidade de Lisboa, Av. Rovisco Pais 1, 1049-001 Lisboa, Portugal\\
$^{2}$ Laboratoire de Physique et Chimie Th\'eoriques, Universit\'e de Lorraine, CNRS, F-54000 Nancy, France\\

\begin{abstract}
The $n$-species particle-exchange process (PEP($n$)) is an exclusion process in which particles of $n$ different species exchange positions on neighbouring sites with rates chosen such that the invariant measure on the discrete torus is a product measure. We address the large-scale hydrodynamic behaviour of this process which yields a
system of $n$ coupled inviscid Burgers 
equations. This system of conservation laws is shown to admit Riemann invariants for arbitrary $n$ from which explicit solutions of the Riemann problem in terms of shock waves and rarefaction fans are obtained. 
We also introduce the open PEP($n$), in which particles are exchanged with boundary reservoirs. For a distinguished manifold of boundary rates, we prove that the invariant measure is the same product measure as in the periodic system. The hydrodynamic description in terms of Riemann invariants is used
to derive the stationary phase diagram explicitly in terms of microscopic boundary rates. In the generic case, the steady state exhibits $2n+1$ phases, with boundary-induced phase transitions analogous to those of the single-species asymmetric simple exclusion process.
\end{abstract}

\section{Introduction}

One-dimensional stochastic interacting particle systems with open
boundaries form a paradigmatic class of nonequilibrium models. In the
bulk, the microscopic dynamics conserves particle numbers locally, while
the boundaries exchange particles with external reservoirs. The stationary
state is therefore driven by the reservoirs and, in the presence of a bias,
typically carries a nonzero current. A characteristic feature of such
systems is that the dependence of the stationary bulk state on the
reservoir parameters may be discontinuous, even for short-range dynamics.
These discontinuities have no equilibrium analogue and became known as
boundary-induced phase transitions \cite{Krug91}.
Beyond their role as exactly tractable models, exclusion-type dynamics
also provide minimal descriptions of collective transport phenomena in
biology, molecular motors, and traffic-like systems \cite{Chow24}.

For systems with a single conserved quantity, the hydrodynamic mechanism
behind these transitions is by now well understood. On a microscopic scale
it has its origin in the flow of local perturbations of the stationary distribution
\cite{Popk99,Hage01}. On the Euler scale it can be understood from the
density evolution according to a scalar conservation law whose flux is the
stationary current-density relation of the corresponding conservative
particle system. With a regularization by the vanishing viscosity method
\cite{Chen95} the stationary bulk density selected by two reservoirs can then be
read off by the so-called
extremal-current principle \cite{Popk99}. For the asymmetric simple
exclusion process (ASEP), originally introduced in connection with
ribosome motion \cite{MacD68} and later developed as a basic model of interacting
particle systems \cite{Spit70}, this gives the familiar
low-density, high-density, and maximal-current phases
\cite{Derr93a,Schu93b,Sand94b}. The same scalar setting also provides
exact microscopic descriptions of shocks and their fluctuations
\cite{Derr93b,Ferr94a,Schu23a}.

The situation is much less complete for systems with several conserved
quantities. In multicomponent systems the hydrodynamic modes interact, and
the reservoirs can no longer be treated through a scalar selection
principle. Boundary effects may be transmitted through different
characteristic families, producing a substantially richer phase structure
\cite{Popk04a,Popk11}. Moreover, even the conservative problem is usually
difficult: a generic system of conservation laws does not admit global
Riemann variables, and for generic multispecies exclusion processes the
invariant measure is not known even for periodic boundary conditions.
As a result, the current-density relation, the
Riemann problem, and the connection between microscopic boundary rates and
macroscopic reservoir densities are rarely all explicit.

Recent work has shown that, for boundary-driven multispecies systems, the
hydrodynamic selection problem can be formulated in terms of Riemann
invariants rather than by a scalar extremal-current principle
\cite{Popk04,Cant24,zahra2025hydrodynamics}. However, in most examples the
phase diagram is obtained only in terms of macroscopic reservoir
densities. Relating these densities explicitly to the microscopic
injection, extraction, and exchange rates at the boundaries remains a
separate and often nontrivial step. Thus, explicit open-boundary phase 
diagrams for systems with several
conserved quantities are still available only in special cases, including
two-component driven systems with exact shock measures or steady-state
selection mechanisms and special open multispecies exclusion processes
\cite{Rako04,Brza07,Ayye17,Roy21}.

In this paper we carry out this program for the $n$-species
particle-exchange process (PEP($n$)) introduced in \cite{Schu17}
for general $n$ and in \cite{Toth03} for $n=2$.
This is a
multispecies exclusion process in which neighbouring particles exchange
positions with species-dependent and direction-dependent rates. We focus on
the parameter regime where the bias of an exchange between two species is
generated by species drifts: each species is assigned a real drift parameter,
and the antisymmetric part of the exchange rate is the difference of the
two drifts, Figure~\ref{fig:PEP($n$)-illustration}. In this regime the process with periodic boundary conditions
admits a factorized invariant measure. This product structure makes it
possible to compute the stationary current-density relation explicitly and
to derive a closed Euler-scale system for the $n$ conserved species.

Our first main result is the complete analysis of this hydrodynamic system
for arbitrary $n$. In the nondegenerate case, where the drift parameters
 are pairwise distinct, we construct global Riemann
variables on the physical simplex. In these variables the equations
diagonalize, the characteristic speeds are ordered, and the Riemann problem
can be solved explicitly for arbitrary left and right states. Each
characteristic family gives rise to Burgers-type elementary waves:
rarefaction fans or Lax shocks. We also describe the degenerate case, in
which several drift parameters coincide and the corresponding species
combine into block densities transported together with contact modes.

Our second main result concerns open boundaries. We introduce boundary
exchange dynamics with reservoirs and identify a manifold of boundary
rates that is compatible with the same product measures as in the
periodic system when the two reservoirs coincide. These boundary
conditions are \emph{PDE-friendly} in the sense of \cite{Popk04}: they
provide an explicit map from microscopic boundary rates to the
macroscopic reservoir densities entering the hydrodynamic equations.

Combining this boundary construction with the explicit solution of the
Riemann problem allows us to derive the hydrodynamic phase diagram of the open
PEP($n$). The selected stationary bulk state is obtained by solving
the Riemann step initial conditions with densities equal to the boundary densities of the reservoirs and evaluating the
self-similar solution at the origin. In Riemann variables this selection
takes a particularly simple form and leads to a generic phase structure
with $2n+1$ phases: $n+1$ boundary-induced phases and $n$ mixed phases in
which one mode is selected by the bulk. For $n=1$ the construction reduces
to the usual ASEP phase diagram, while the case $n=2$ gives a transparent
genuinely multicomponent example, which we discuss explicitly.

The significance of the PEP($n$) is therefore twofold. On the hydrodynamic
side, it provides a rare example of a genuinely
multicomponent driven system with an arbitrary number $n$ of conservation laws whose Euler equations admit global Riemann
variables and for which Riemann problem is explicitly solvable. On the
nonequilibrium side, it gives an open multispecies exclusion process for
which the phase diagram of boundary-induced phase transitions
can be expressed directly in terms of
the microscopic boundary parameters. The model thus provides an explicit
bridge between microscopic exchange dynamics, macroscopic wave
propagation, and boundary-induced steady-state selection in systems with several
conserved species.

The paper is organized as follows. In Section~\ref{PEP} we define formally the
PEP($n$) with periodic boundaries and recall its product invariant
measure and stationary current-density relation. In Section~\ref{hydro} we derive
the Euler-scale hydrodynamic equations, construct the Riemann variables,
and solve the Riemann problem. In Section~\ref{open} we introduce open boundaries,
identify the PDE-friendly boundary manifold, and derive the phase diagram
for boundary-induced phase transitions, including the two-species example.

\begin{figure}[h!]
\centering
\begin{tikzpicture}[
  site/.style    ={circle, draw=black, thick, minimum size=8mm, inner sep=0pt, fill=white},
  particle/.style={circle, draw=black, thick, minimum size=8mm, inner sep=0pt, font=\small\bfseries},
  spA/.style={particle, fill=red!70,         text=white},
  spB/.style={particle, fill=blue!65,        text=white},
  spC/.style={particle, fill=green!55!black, text=white},
]

\draw[gray!55, thick] (-0.2,0) -- (10.2,0);
\node at (-0.55,0) {$\cdots$};
\node at (10.55,0) {$\cdots$};

\node[site] at (1,0)  {};
\node[spA]  at (2,0)  {$1$};
\node[spB]  at (3,0)  {$2$};
\node[site] at (4,0)  {};
\node[spC]  at (5,0)  {$3$};
\node[spA]  at (6,0)  {$1$};
\node[site] at (7,0)  {};
\node[spB]  at (8,0)  {$2$};
\node[spC]  at (9,0)  {$3$};
\node[site] at (10,0) {};

\draw[->, thick] (2,0.5) to[bend left=55]
   node[midway, above, font=\small] {$g_{1,2}$}
   (3,0.5);

\draw[->, thick] (6,0.5) to[bend left=55]
   node[midway, above, font=\small] {$g_{1,0}$}
   (7,0.5);

\node[font=\footnotesize, gray!70!black] at (2,-0.7)  {$k$};
\node[font=\footnotesize, gray!70!black] at (3,-0.7)  {$k{+}1$};
\node[font=\footnotesize, gray!70!black] at (6,-0.7)  {$\ell$};
\node[font=\footnotesize, gray!70!black] at (7,-0.7)  {$\ell{+}1$};

\node[
  draw=gray!70, rounded corners, thick,
  fill=gray!8,
  inner sep=6pt,
  align=left,
  font=\small,
  anchor=north,
] at (5,-1.4) {%
  \begin{tabular}{@{}r@{\;}l@{}}
   Exchange rate: & $g_{\alpha,\beta}\;=\;\tfrac{1}{2}\bigl(w_{\alpha,\beta} + f_\alpha - f_\beta\bigr)$, \quad
                    $\alpha,\beta\in\{0,1,\dots,n\}$ \\[2pt]
   Symmetric part:& $w_{\alpha,\beta}=w_{\beta,\alpha}>0$ for $\alpha\neq\beta$, \;
                    $w_{\alpha,\alpha}=0$ \\[2pt]
   Drifts:        & $f_\alpha\in\mathbb{R}$, \;
                    $f_0=0$ (vacancies are passive) \\[2pt]
   Positivity:    & $w_{\alpha,\beta}\;\geq\;|f_\alpha-f_\beta|$ \\
  \end{tabular}
};

\end{tikzpicture}
\caption{The bulk dynamics of the $n$-species particle exchange process, illustrated for $n=3$. The constraints in the box make the invariant measure a product measure.}
\label{fig:PEP($n$)-illustration}
\end{figure}

\section{The $n$-species particle exchange process}
\label{PEP}
\setcounter{equation}{0}
\setcounter{footnote}{0}

Exclusion processes with more than one conserved species of particles
do not in general have a factorized invariant measure even
when interactions are nearest neighbor and boundary conditions 
are periodic, see \cite{Derr93b,Evan98b,Isae01,Ferr13} for some
well-known examples that have attracted considerable attention.
As explored in \cite{Schu17} there is, however, a parameter
manifold of such multispecies exclusion processes which have
an invariant product measure. This family of processes is called
the $n$-species particle exchange process (PEP($n$)), defined
formally below in \eqref{genPEP}.

\subsection{Definition}

We consider the finite lattice with $L$ sites
with periodic boundary conditions, i.e., the one-dimensional
torus $\T_L = \{1,\dots,L\}$. Each site $k\in\T_L$ (counted
modulo $L$) is either empty or occupied by at most one particle
(exclusion principle).
Each particle belongs to a species $\alpha\in\{1,\dots,n\}$.
A particle of species 0 is called vacancy. When $\alpha_{k}=0$ 
we also say that site $k$ is empty. 

A configuration of the PEP is denoted
by the $L$-tuple $\balpha=(\alpha_{1},\alpha_{2},\dots,
\alpha_{L})$ where $\alpha_{k}\in\{0,\dots,n\}$
indicates the presence of a particle of type $\alpha_k$ on site $k$.
The indicator
\begin{equation}
n^{\alpha}_{k} = \delta_{\alpha_{k},\alpha}
\end{equation}
``counts'' the number of particles of type $\alpha$ on site $k$.
The total number of particles of each species 
in the system is given by
\begin{equation}
N_{\alpha} = \sum_{k=1}^{L} n^{\alpha}_{k} .
\label{partnum}
\end{equation}
Notice that by the exclusion principle the state space $\Omega$ 
of the PEP($n$) is the set of
all elements in $\{0,\dots,n\}^L$ subject to the condition that
$\sum_{\alpha=0}^n n^{\alpha}_{k}=1$
for all $k\in\T_L$.

With the swapped configuration
$\balpha^{k,l}$ with occupation variables
\begin{equation}
\alpha^{k,l}_m := \left\{\ba{cc} 
\alpha_{l} & \mbox{ if } m=k \\
\alpha_{k} & \mbox{ if } m=l \\
\alpha_{m} & \mbox{ else } 
\ea\right.
\end{equation}
the local exchange generators $\mathcal{L}^{ex}_{k,k+1}$ defined by
\begin{equation}
(\mathcal{L}^{ex}_{k,k+1} f) (\balpha) = 
\sum_{\alpha=0}^n \sum_{\beta=0}^n 
g_{\alpha, \beta} n^{\alpha}_{k} n^{\beta}_{k+1}
\left[f(\balpha^{k,k+1}) - f(\balpha)\right] 
\label{genPEPloc}
\end{equation}
for measurable functions $f:\Omega\to\R$
encode the exchange of
an exclusion particle of species
$\alpha \in \{0,1,\dots,n\}$ on site $k\in\T_L$ with an 
exclusion particle of species
$\beta$ on site $k+1 \mod{L}$ with rate $g_{\alpha,\beta}$,
i.e., after an exponentially distributed independent random time with
mean $1/g_{\alpha,\beta}$. 
The Markovian dynamics of the PEP with periodic boundary conditions 
is then defined by the sum
\begin{equation}
\mathcal{L}^{per} f = \sum_{k=1}^{L} \mathcal{L}^{ex}_{k,k+1} 
\label{genPEP}
\end{equation}
of local exchange  generators.

It is convenient to decompose the rates into a symmetric part 
$w_{\alpha,\beta}=w_{\beta,\alpha}>0$ for $\alpha\neq\beta$, 
$w_{\alpha,\alpha}=0$, 
and an antisymmetric part $f_{\alpha,\beta}=-f_{\beta,\alpha}$ 
in the form
\bel{rates}
g_{\alpha, \beta} = \frac{1}{2} (w_{\alpha, \beta} + f_{\alpha,\beta}) .
\ee
Strict positivity of the rates implies $w_{\alpha,\beta} \geq |f_{\alpha,\beta}|$ and antisymmetry implies  
$f_{\alpha,\alpha} = 0$. For the exchange rates of particles of 
species
$\alpha\in\{1,\dots,n\}$ with vacant neighbors the short-hand
notations
\bel{Edef}
f_\alpha := f_{\alpha,0}, \quad w_\alpha := w_{\alpha,0}
\ee
is used. Notice that $f_{0}=w_{0}=0$. 

\subsection{Invariant measure}

Due to the exchange dynamics the state space decomposes into 
ergodic components specified by the particle numbers $N_\alpha$.
If the condition
\bel{stat}
f_{\alpha,\beta} = f_\alpha - f_\beta
\ee
holds, then the unique invariant measure in such an ergodic
component is the
uniform distribution \cite{Schu17}.
It follows that with the fugacities $\mu_\alpha\in\R$ the grand canonical product measure
\begin{equation}
\pi = \prod_{k=1}^{L} \left(\sum_{\alpha=0}^n \frac{\rme^{\mu_\alpha} n_k^{\alpha}}{\sum_{\alpha=0}^n \rme^{\mu_\alpha}} \right)
\label{invmeas}
\end{equation}
with fluctuating particle numbers is also stationary. 
Expectations w.r.t. this invariant measure are denoted by
$\exval{\cdot}$. Using standard methods \cite{Ligg99} one can show that
the (infinite) product measure with these marginals is the
unique and extremal translation-invariant stationary measure
of the PEP defined on the infinite integer lattice $\Z$.

For this measure the particle densities $\rho_\alpha := \exval{N_{\alpha}}/L $ are given by
\be
\rho_\alpha
= \frac{\rme^{\mu_\alpha} }{ \sum_{\beta=0}^n \rme^{\mu_{\beta}}}.
\ee
Exclusion implies $\rho_0 = 1 - \sum_{\alpha=1}^n \rho_\alpha$.
For the covariances (generalized compressibilities)
\be
\kappa_{\alpha \beta} := \frac{\partial \rho_\alpha}{\partial \mu_\beta} =
\frac{1}{L} \exval{(N_{\alpha}-\exval{N_{\alpha}})(N_{\beta}-\exval{N_{\beta}})} 
\ee
one obtains due to the product structure
\be
\kappa_{\alpha \beta} = \rho_\alpha 
(\delta_{\alpha,\beta}-\rho_\beta).
\ee
The compressibility
matrix with matrix elements $\kappa_{\alpha\beta}$ is denoted by $K$. By construction $K=K^T$ is symmetric.

\subsection{Stationary current-density relation}

Particle conservation yields the discrete continuity equation
\begin{equation}
\mathcal{L}^{per} n^\alpha_{k} = j_{k-1}^\alpha - j_k^\alpha
\label{cebulk}
\end{equation}
with the instantaneous currents 
\be
 j_k^\alpha = \sum_{\beta=0}^n g_{\alpha,\beta} n_k^\alpha n_{k+1}^\beta - g_{\beta,\alpha} n_k^\beta n_{k+1}^\alpha.
\label{instcurr}
\ee
The  factorization property of the
grand canonical stationary distribution yields the stationary current-density relation
\bel{statcur}
j_\alpha := \exval{ j_k^\alpha} = \rho_\alpha 
\left(f_\alpha- \sum_{\beta=1}^n f_\beta \rho_\beta\right).
\ee
The flux Jacobian $J$ with matrix elements
\bel{velmat}
J_{\alpha\beta} = \frac{\partial j_\alpha}{\partial \rho_\beta} = 
\left[ f_\alpha  - \sum_{\gamma=1}^n f_{\gamma} \rho_\gamma\right]  \delta_{\alpha,\beta}-  f_\beta \rho_\alpha
\ee
satisfies the Onsager-type current symmetry $JK = (JK)^T$ \cite{Toth03,Gris11}. The eigenvalues $v_\alpha$ of $J$ 
are the collective velocities of the normal modes
which are the center of mass velocities of localized
perturbations of the stationary distribution \cite{Schu03}.

\section{Hydrodynamic behavior of the PEP($n$)}
\label{hydro}

We consider the conservative PEP($n$) on the infinite lattice and its
Euler-scale hydrodynamic description. For one-dimensional interacting particle
systems with finitely many local states and several conserved quantities, a
general hydrodynamic-limit theory was developed by T\'oth and Valk\'o
\cite{Toth03}. In the present work, we use this general framework as motivation for the
macroscopic equations and their explicit solution below.

Let $\rho_\alpha(x,t)$, $\alpha=1,\dots,n$, denote the macroscopic density of
species $\alpha$ at position $x\in\R$ and time $t\geq 0$. Since the microscopic
dynamics conserves the number of particles of each species, the Euler-scale
evolution is governed by $n$ coupled Burgers equations, viz., the system conservation laws
\bel{hydropep}
\partial_t \rho_\alpha(x,t)+\partial_x j_\alpha(\brho(x,t))=0,
\qquad \alpha=1,\dots,n,
\ee
where $\brho=(\rho_1,\dots,\rho_n)$ and the fluxes are given by the stationary
current-density relation \eref{statcur}.
Introducing the mean drift
\bel{meandrift}
U(\brho):=\sum_{\beta=1}^n f_\beta \rho_\beta,
\ee
the hydrodynamic equations take the equivalent form
\bel{hydroU}
\partial_t \rho_\alpha+\partial_x\Bigl[\rho_\alpha(f_\alpha-U)\Bigr]=0,
\qquad \alpha=1,\dots,n.
\ee

Because of the exclusion constraint, the vacancy density is not an independent
field but is given by
\be
\rho_0(x,t)=1-\sum_{\beta=1}^n \rho_\beta(x,t).
\ee
Since we have $f_0=0$, then $\rho_0$ satisfies the same conservation law,
\be
\partial_t \rho_0+\partial_x\Bigl[\rho_0(f_0-U)\Bigr]=0,
\ee
so that the full system can be written in the symmetric form
\bel{fullhydro}
\partial_t \rho_\alpha+\partial_x\Bigl[\rho_\alpha(f_\alpha-U)\Bigr]=0,
\qquad \alpha=0,1,\dots,n,
\ee
together with the algebraic constraint
\be
\sum_{\alpha=0}^n \rho_\alpha=1.
\ee
Thus, the natural state space of the hydrodynamic system is the simplex
\be
\mathcal{D}
=
\left\{
(\rho_1,\dots,\rho_n)\in \R^n:
\rho_\alpha\geq 0,\;
\sum_{\alpha=1}^n \rho_\alpha\leq 1
\right\},
\ee
with interior
\be
\mathring{\mathcal{D}}
=
\left\{
(\rho_1,\dots,\rho_n)\in \R^n:
\rho_\alpha>0,\;
\sum_{\alpha=1}^n \rho_\alpha<1
\right\}.
\ee
In the following subsection we show that, for states in
$\mathring{\mathcal{D}}$, the system \eref{hydropep} admits global Riemann
variables, which allow us to solve the associated Riemann problem explicitly.

\subsection{The Riemann invariants}
We now introduce a set of variables which diagonalize the hydrodynamic system
\eref{fullhydro}. Throughout this subsection we assume that the $n+1$ numbers $f_0,f_1,\dots,f_n$
are pairwise distinct. This is the non-degenerate case. If two of them coincide, the construction below degenerates and strict hyperbolicity may fail.
To a state $\rho\in\mathring{\mathcal D}$ we associate the rational function
\begin{equation}
R_n(\omega) := \sum_{\alpha=0}^n \frac{\rho_\alpha}{f_\alpha-\omega}.
\label{eq:Rmu}
\end{equation}
This function has a simple pole at each $\omega=f_\alpha$. Let $\sigma$ be the permutation of $\{0,1,\dots,n\}$ such that
\[
 f_{\sigma(0)} < f_{\sigma(1)} < \cdots < f_{\sigma(n)}.
\]
The ordering is only used to identify the intervals between consecutive pole locations; the original notation $\rho_\alpha,f_\alpha$ is kept throughout.

\begin{proposition}[Existence of one root in each gap]
For every state $\rho\in\mathring{\mathcal D}$, the function $R_n$ has exactly one simple zero in each interval
\[
\bigl(f_{\sigma(i-1)},f_{\sigma(i)}\bigr),
\qquad i=1,\dots,n.
\]
\end{proposition}

\begin{proof}
Differentiate \eqref{eq:Rmu} with respect to $\omega$:
\[
R_n'(\omega)=\sum_{\alpha=0}^n \frac{\rho_\alpha}{(f_\alpha-\omega)^2}>0
\]
whenever $\omega\neq f_\alpha$ for all $\alpha$. Hence $R_n$ is strictly increasing on each connected component of
\[
\R\setminus\{f_0,\dots,f_n\}.
\]
Now fix $i\in\{1,\dots,n\}$. As $\omega\downarrow f_{\sigma(i-1)}$, the term with index $\sigma(i-1)$ dominates and
$R_n(\omega)\to -\infty.$
As $\omega\uparrow f_{\sigma(i)}$, the term with index $\sigma(i)$ dominates and
$R_n(\omega)\to +\infty.$
Since $R_n$ is continuous and strictly increasing on the interval
$\bigl(f_{\sigma(i-1)},f_{\sigma(i)}\bigr),$
it has exactly one zero there. The non-degeneracy of the root follows since $R_n'(\omega)>0$ at the zero.
\end{proof}
We denote the unique zero in the interval
$\bigl(f_{\sigma(i-1)},f_{\sigma(i)}\bigr)$ by
\[
z_i=z_i(\rho),\qquad i=1,\dots,n,
\]
so that
\begin{equation}
 f_{\sigma(i-1)}<z_i<f_{\sigma(i)},
 \qquad i=1,\dots,n.
 \label{eq:r-domain}
\end{equation}
We will show next that the numbers $z_1,\dots,z_n$ constitute the Riemann variables.

Because $R_n(\omega)$ has poles at $f_0,\dots,f_n$ and zeros at $z_1,\dots,z_n$, and because
\[
R_n(\omega)\sim -\frac{1}{\omega}
\qquad (\omega\to\infty)
\]
by the identity $\sum_{\alpha=0}^n \rho_\alpha=1$, we obtain the factorization
\begin{equation}
R_n(\omega) = -\frac{\prod_{i=1}^n (\omega-z_i)}{\prod_{\alpha=0}^n (\omega-f_\alpha)}.
\label{eq:R-factorized}
\end{equation}
We show next that this gives an explicit inverse formula for the densities.
\begin{proposition}[Inverse map]
For each $\alpha=0,1,\dots,n$,
\begin{equation}
\rho_\alpha
=
\frac{\prod_{i=1}^n (f_\alpha-z_i)}{\prod_{\beta\ne\alpha}(f_\alpha-f_\beta)}.
\label{eq:inverse-density}
\end{equation}
Therefore the map
\[
\rho\longmapsto (z_1,\dots,z_n)
\]
is one-to-one on $\mathring{\mathcal D}$.
\end{proposition}

\begin{proof}
Take residues at the pole $\omega=f_\alpha$. From \eqref{eq:Rmu}, the residue is $-\rho_\alpha$ because
\[
\frac{\rho_\alpha}{f_\alpha-\omega}=-\frac{\rho_\alpha}{\omega-f_\alpha}.
\]
From \eqref{eq:R-factorized}, the residue is
\[
-\frac{\prod_{i=1}^n (f_\alpha-z_i)}{\prod_{\beta\ne\alpha}(f_\alpha-f_\beta)}.
\]
Equating the two expressions gives \eqref{eq:inverse-density}.
\end{proof}

The inequalities \eqref{eq:r-domain} define the natural domain of the Riemann variables:
\[
\mathcal D_z
:=
\left\{(z_1,\dots,z_n)\in\R^n:\ f_{\sigma(i-1)}<z_i<f_{\sigma(i)},\ i=1,\dots,n\right\}.
\]
Since the inverse map is explicit, the map $\rho\mapsto z$ is a smooth bijection between $\mathring{\mathcal D}$ and $\mathcal D_z$.

\begin{remark}
The formula \eqref{eq:inverse-density} automatically produces positive densities when the variables $z_i$ satisfy the interlacing inequalities \eqref{eq:r-domain}.
\end{remark}

We now derive the PDE satisfied by $R_n(\omega)$. Since
\[
\p_t\rho_\alpha = -\p_x\bigl(\rho_\alpha(f_\alpha-U)\bigr)
=-(f_\alpha-U)\p_x\rho_\alpha + \rho_\alpha\,\p_x U,
\]
we obtain
\begin{align*}
\p_t R_n(\omega)
&=
\sum_{\alpha=0}^n \frac{\p_t\rho_\alpha}{f_\alpha-\omega}
\\
&=
-\sum_{\alpha=0}^n \frac{(f_\alpha-U)\p_x\rho_\alpha}{f_\alpha-\omega}
+ (\p_x U)\sum_{\alpha=0}^n \frac{\rho_\alpha}{f_\alpha-\omega}.
\end{align*}
Now use
\[
\frac{f_\alpha-U}{f_\alpha-\omega}=1+\frac{\omega-U}{f_\alpha-\omega},
\]
so that
\begin{align*}
\sum_{\alpha=0}^n \frac{(f_\alpha-U)\p_x\rho_\alpha}{f_\alpha-\omega}
&=
\sum_{\alpha=0}^n \p_x\rho_\alpha
+
(\omega-U)\sum_{\alpha=0}^n \frac{\p_x\rho_\alpha}{f_\alpha-\omega}
\\
&=
(\omega-U)\,\p_x R_n(\omega),
\end{align*}
because $\sum_{\alpha=0}^n\rho_\alpha=1$ implies $\sum_{\alpha=0}^n\p_x\rho_\alpha=0$. We conclude that
\begin{equation}
\p_t R_n(\omega) + (\omega-U)\p_x R_n(\omega) - (\p_x U)R_n(\omega)=0.
\label{eq:R-pde}
\end{equation}


We now evaluate \eqref{eq:R-pde} at a zero $\omega=z_i(x,t)$ of $R_n$. Since
\[
R_n(z_i(x,t),x,t)=0,
\]
we have, by the chain rule,
\[
\p_t R_n(z_i)+R_n'(z_i)\,\p_t z_i = 0,
\qquad
\p_x R_n(z_i)+R_n'(z_i)\,\p_x z_i = 0.
\]
Substituting into \eqref{eq:R-pde} at $\omega=z_i$ gives
\[
-R_n'(z_i)\,\p_t z_i -(z_i-U)R_n'(z_i)\,\p_x z_i = 0.
\]
Since $R_n'(z_i)>0$, we obtain the diagonal system
\begin{equation}
\p_t z_i + \lambda_i(z)\,\p_x z_i = 0,
\qquad i=1,\dots,n,
\label{eq:diagonal-system}
\end{equation}
with characteristic speeds
\begin{equation}
\lambda_i(z)=z_i-U.
\label{eq:lambda-ri-minus-U}
\end{equation}
This proves that the variables $z_1,\dots,z_n$ are Riemann variables for the system.

To make the speeds completely explicit, we expand \eqref{eq:R-factorized} at infinity. On the one hand,
\[
R_n(\omega) = -\frac{1}{\omega} - \frac{U}{\omega^2}+O(\omega^{-3}).
\]
On the other hand,
\[
R_n(\omega)
=
-\frac{1}{\omega}
\frac{\prod_{i=1}^n (1-z_i/\omega)}{\prod_{\alpha=0}^n (1-f_\alpha/\omega)}
=
-\frac{1}{\omega}
\left(1+\frac{\sum_{\alpha=0}^n f_\alpha-\sum_{i=1}^n z_i}{\omega}+O(\omega^{-2})\right).
\]
Comparing the coefficients of $\omega^{-2}$ yields
\begin{equation}
U = \sum_{\alpha=0}^n f_\alpha - \sum_{i=1}^n z_i.
\label{eq:U-in-r}
\end{equation}
Therefore
\begin{equation}
\lambda_i(z)
=
z_i + \sum_{k=1}^n z_k - \sum_{\alpha=0}^n f_\alpha,
\qquad i=1,\dots,n.
\label{eq:explicit-speeds}
\end{equation}

\begin{theorem}[Global Riemann variables]
\label{thm:global-riemann-variables}
Assume that $f_0=0,f_1,\dots,f_n$ are pairwise distinct. Then on the interior of the physical simplex, the variables $z_1,\dots,z_n$ defined as the zeros of
\[
R_n(\omega)=\sum_{\alpha=0}^n \frac{\rho_\alpha}{f_\alpha-\omega}
\]
provide a global change of variables which diagonalizes the system. More precisely:
\begin{enumerate}[label=\textup{(\roman*)}]
\item for each state there is exactly one $z_i$ in each gap between two consecutive ordered values of $f_0,\dots,f_n$;
\item the inverse map is given by \eqref{eq:inverse-density};
\item the PDE becomes the diagonal system \eqref{eq:diagonal-system} with speeds \eqref{eq:explicit-speeds}.
\end{enumerate}
\end{theorem}

\begin{corollary}[Strict hyperbolicity]
In the interior of the simplex,
\[
\lambda_j(z)-\lambda_i(z)=z_j-z_i,
\qquad i\ne j.
\]
Since the variables $z_i$ lie in disjoint ordered intervals, one has $z_1<z_2<\cdots<z_n$, hence
\[
\lambda_1<\lambda_2<\cdots<\lambda_n.
\]
Therefore the system is strictly hyperbolic on $\mathring{\mathcal D}$.
\end{corollary}

\begin{corollary}[Genuine nonlinearity]
Each characteristic field is genuinely nonlinear, because
\[
\frac{\partial \lambda_i}{\partial z_i}=2,
\qquad i=1,\dots,n.
\]
\end{corollary}

\begin{remark}[Semi-Hamiltonian structure]
The diagonal system belongs to the class of semi-Hamiltonian systems
of hydrodynamic type in the sense of Tsarev, also called rich systems
in the conservation-law literature \cite{Tsarev,Dubrovin,Serre}. Indeed, for
$i\neq j$,
\[
\frac{\partial_{z_j}\lambda_i}{\lambda_j-\lambda_i}
=
\frac{1}{z_j-z_i}.
\]
Hence, for distinct $i,j,k$,
\[
\partial_{z_k}
\left(
\frac{\partial_{z_j}\lambda_i}{\lambda_j-\lambda_i}
\right)
=
\partial_{z_j}
\left(
\frac{\partial_{z_k}\lambda_i}{\lambda_k-\lambda_i}
\right)
=0.
\]
Thus Tsarev's compatibility conditions are satisfied. Consequently,
the system falls within the class to which Tsarev's generalized
hodograph method applies \cite{Tsarev}. This observation is not needed
for the self-similar Riemann solutions constructed below, but it places
the model in the standard integrable class of diagonal hydrodynamic-type
systems.
\end{remark}

\subsection{Riemann problem in the $z$-variables:}

We now consider the Riemann problem, which consists of the system of conservation laws \eref{hydropep} defined in the real line with initial data given by:
\begin{equation}
\rho(x,0)=
\begin{cases}
\rho^L,& x<0,\\
\rho^R,& x>0,
\end{cases}
\label{eq:Riemann-data-rho}
\end{equation}
where the left and right states belong to the interior of the simplex.
Let
\[
\mathbf{z}^L=(z_1^L,\dots,z_n^L),
\qquad
\mathbf{z}^R=(z_1^R,\dots,z_n^R)
\]
be the corresponding Riemann variables. Because the map $\rho\mapsto z$ is invertible, solving the Riemann problem in $\rho$ is equivalent to solving it in $z$.
We recall from the theory of hyperbolic conservation laws
\cite{Lax73,Lax00,Serr00,Serre} that the solution of the Riemann problem
is built from elementary waves, each associated with one of the $n$
characteristic families. For introductory reviews, see \cite[Ch.~1]{ZahraMulti} and
\cite{zahra2025hydrodynamics}. These waves are spatially separated, with the natural intermediate states given by:
\begin{equation}
\mathbf{z}^{(k)}:=\bigl(z_1^R,\dots,z_k^R,z_{k+1}^L,\dots,z_n^L\bigr),
\qquad k=0,1,\dots,n,
\label{eq:intermediate-states}
\end{equation}
with the conventions
\[
\mathbf{z}^{(0)}=\mathbf{z}^L,
\qquad
\mathbf{z}^{(n)}=\mathbf{z}^R.
\]
The $k$-th wave connects $\mathbf{z}^{(k-1)}$ to $\mathbf{z}^{(k)}$, so only the variable $z_k$ changes across that wave.

\paragraph{Reduction of the $k$-th family to a scalar conservation law:}
Fix $k\in\{1,\dots,n\}$ and freeze all the other Riemann variables at
\[
z_j=z_j^R\quad (j<k),
\qquad
z_j=z_j^L\quad (j>k).
\]
Along the $k$-th wave only $z_k$ varies. Define the constant
\begin{equation}
C_k:=\sum_{j<k} z_j^R + \sum_{j>k} z_j^L - \sum_{\alpha=0}^n f_\alpha.
\label{eq:Ck}
\end{equation}
Then along this wave,
\begin{equation}
\lambda_k=2z_k+C_k.
\label{eq:lambda-k-reduced}
\end{equation}
Hence the equation for $z_k$ reduces to the scalar conservation law
\begin{equation}
\p_t z_k + \p_x g_k(z_k)=0,
\qquad
g_k(q)=q^2+C_k q,
\label{eq:scalar-reduction}
\end{equation}
up to an irrelevant additive constant in the flux, this is simply the Burgers' equation.
Since
\[
g_k''(q)=2>0,
\]
the reduced scalar flux is strictly convex.

Therefore the $k$-th wave is:
\begin{itemize}
\item a rarefaction if $z_k^L<z_k^R$;
\item a shock if $z_k^L>z_k^R$;
\item absent if $z_k^L=z_k^R$.
\end{itemize}
\paragraph{The $k$-th rarefaction fan:}Assume first that
\[
z_k^L<z_k^R.
\]
Then the $k$-th wave is a centered rarefaction fan. Its left and right edge speeds are
\begin{equation}
\lambda_k^-:=\lambda_k\bigl(\mathbf{z}^{(k-1)}\bigr)=2z_k^L+C_k,
\qquad
\lambda_k^+:=\lambda_k\bigl(\mathbf{z}^{(k)}\bigr)=2z_k^R+C_k.
\label{eq:rarefaction-edges}
\end{equation}
Inside the fan, the self-similar solution $z(x,t)=\mathbf{z}(\xi)$ with $\xi=x/t$ is given by
\begin{equation}
z_j(\xi)=
\begin{cases}
z_j^R,& j<k,\\
\displaystyle \frac{\xi-C_k}{2},& j=k,\\
z_j^L,& j>k,
\end{cases}
\qquad \lambda_k^-\le \xi\le \lambda_k^+.
\label{eq:rarefaction-profile}
\end{equation}
Outside that interval, $z_j$ takes the constant values according to $\mathbf{z}^{(k-1)}$ and $\mathbf{z}^{(k)}$.

\paragraph{The $k$-th shock wave:}Assume now that
\[
z_k^L>z_k^R.
\]
Then the $k$-th wave is a Lax shock. Its speed is the Rankine--Hugoniot speed for the scalar flux \eqref{eq:scalar-reduction}:
\begin{equation}
s_k
=
\frac{g_k(z_k^R)-g_k(z_k^L)}{z_k^R-z_k^L}
=
z_k^L+z_k^R+C_k.
\label{eq:shock-speed}
\end{equation}
The entropy condition is automatic because the flux is convex:
\[
\lambda_k\bigl(\mathbf{z}^{(k)}\bigr)=2z_k^R+C_k
< s_k <
2z_k^L+C_k=\lambda_k\bigl(\mathbf{z}^{(k-1)}\bigr).
\]
So the $k$-th wave is the discontinuity
\begin{equation}
\mathbf{z}(\xi)=
\begin{cases}
\mathbf{z}^{(k-1)},& \xi<s_k,\\
\mathbf{z}^{(k)},& \xi>s_k.
\end{cases}
\label{eq:shock-profile}
\end{equation}

\begin{remark}
Although the variables $z_i$ are not conservative variables, the formula \eqref{eq:shock-speed} is consistent with the Rankine--Hugoniot conditions for the original conservative system in $\rho$. Indeed, along a $k$-wave all densities are affine functions of $z_k$ by \eqref{eq:inverse-density}, so each current $j_\alpha$ becomes a quadratic polynomial in $z_k$ with the same scalar shock speed.
\end{remark}

\subsubsection*{The full entropy solution of the Riemann problem}

We can now write the full self-similar solution.

\begin{theorem}[Riemann solution in Riemann variables]
Let $\rho^L,\rho^R$ be two states in the interior of the simplex, and let $\mathbf{z}^L,\mathbf{z}^R$ be the corresponding Riemann variables. Then the entropy solution of the Riemann problem \eqref{eq:Riemann-data-rho} is the self-similar function
\[
\mathbf{z}(x,t)=\mathbf{z}(\xi),\qquad \xi=\frac{x}{t},
\]
obtained by concatenating $n$ simple waves in the order $1,2,\dots,n$.

More precisely, define the intermediate states by \eqref{eq:intermediate-states}. For each $k=1,\dots,n$, the $k$-th wave connects $\mathbf{z}^{(k-1)}$ to $\mathbf{z}^{(k)}$ and is:
\begin{itemize}
\item a rarefaction if $z_k^L<z_k^R$, given by \eqref{eq:rarefaction-profile};
\item a shock if $z_k^L>z_k^R$, given by \eqref{eq:shock-profile} with speed \eqref{eq:shock-speed};
\item no wave if $z_k^L=z_k^R$.
\end{itemize}
The waves are automatically ordered from left to right according to their family index.
\end{theorem}
\begin{proof}
The reduction of each family to the scalar convex conservation law
\eqref{eq:scalar-reduction} shows that the elementary $k$-wave is exactly
the scalar entropy solution for that flux. It remains to check that these
elementary waves are ordered from left to right according to the family index.

Fix $k\in\{1,\dots,n-1\}$. We compare the rightmost speed of the $k$-wave
with the leftmost speed of the $(k+1)$-wave.

By Theorem~\ref{thm:global-riemann-variables} (i), each $z_i$ always lies in the $i$-th fixed gap
between consecutive ordered values of $f_0,\dots,f_n$. Therefore
\[
\max\{z_k^L,z_k^R\}<\min\{z_{k+1}^L,z_{k+1}^R\}.
\]

For the $k$-wave, the largest speed is
\[
\mu_k^+=
\begin{cases}
2z_k^R+C_k, & z_k^L<z_k^R \quad\text{(rarefaction)},\\[1mm]
z_k^L+z_k^R+C_k, & z_k^L>z_k^R \quad\text{(shock)}.
\end{cases}
\]
Equivalently,
\[
\mu_k^+=C_k+z_k^R+\max\{z_k^L,z_k^R\}.
\]

Next, from \eqref{eq:Ck},
\[
C_{k+1}=C_k+z_k^R-z_{k+1}^L.
\]
For the $(k+1)$-wave, the smallest speed is
\[
\mu_{k+1}^-=
\begin{cases}
2z_{k+1}^L+C_{k+1}, & z_{k+1}^L<z_{k+1}^R \quad\text{(rarefaction)},\\[1mm]
z_{k+1}^L+z_{k+1}^R+C_{k+1}, & z_{k+1}^L>z_{k+1}^R \quad\text{(shock)}.
\end{cases}
\]
Equivalently,
\[
\mu_{k+1}^-=C_k+z_k^R+\min\{z_{k+1}^L,z_{k+1}^R\}.
\]

Using the gap inequality above, we obtain
\[
\mu_k^+<\mu_{k+1}^-.
\]
Thus the whole $k$-wave lies strictly to the left of the whole $(k+1)$-wave.
Since this holds for every $k=1,\dots,n-1$, the waves can be concatenated in
the order $1,2,\dots,n$ without overlap.
\end{proof}

Once the self-similar solution $\mathbf{z}(\xi)$ has been obtained, the density profile is reconstructed explicitly by
\begin{equation}
\rho_\alpha(\xi)
=
\frac{\prod_{i=1}^n \bigl(f_\alpha-z_i(\xi)\bigr)}{\prod_{\beta\ne\alpha}(f_\alpha-f_\beta)},
\qquad \alpha=0,1,\dots,n.
\label{eq:rho-reconstruction}
\end{equation}

\subsection{Degenerate values of the parameters}

So far we have assumed that the parameters $f_0,\dots,f_n$ are pairwise
distinct. We now briefly describe what changes when this assumption fails.
Assume still that $f_0=0$, but allow repetitions among the values
$f_\alpha$.

Let
\[
c_0<c_1<\dots<c_m
\]
be the distinct values taken by $f_0,\dots,f_n$, and let
\[
G_a:=\{\alpha\in\{0,\dots,n\}: f_\alpha=c_a\},
\qquad
\nu_a:=|G_a|.
\]
Define the total density of each degenerate block by
\[
\eta_a:=\sum_{\alpha\in G_a}\rho_\alpha.
\]
Summing \eqref{fullhydro} over $\alpha\in G_a$ gives
\[
\partial_t \eta_a+\partial_x\bigl(\eta_a(c_a-U)\bigr)=0,
\qquad
U=\sum_{a=0}^m c_a\,\eta_a.
\]
Thus the block totals $\eta_0,\dots,\eta_m$ satisfy exactly the same
system as before, but now with the \emph{distinct} parameters
$c_0,\dots,c_m$.

The rational function \eqref{eq:Rmu} becomes
\[
R_n(\omega)
=
\sum_{\alpha=0}^n \frac{\rho_\alpha}{f_\alpha-\omega}
=
\sum_{a=0}^m \frac{\eta_a}{c_a-\omega}.
\]
Hence only the distinct values $c_a$ matter: $R_n$ has $m+1$ poles and
therefore only $m$ zeros,
\[
c_0<z_1<c_1<\dots<z_m<c_m.
\]
The inverse formula \eqref{eq:inverse-density} is correspondingly replaced
by the reconstruction of the block totals,
\[
\eta_a
=
\frac{\prod_{i=1}^m (c_a-z_i)}
{\prod_{b\neq a}(c_a-c_b)}.
\]
So the variables $z_1,\dots,z_m$ determine the sums $\eta_a$, but no
longer determine the individual densities inside a degenerate block.

To describe the remaining freedom, write, whenever $\eta_a>0$,
\[
\theta_{a,\alpha}:=\frac{\rho_\alpha}{\eta_a},
\qquad \alpha\in G_a,
\qquad
\sum_{\alpha\in G_a}\theta_{a,\alpha}=1.
\]
Using the equations for $\rho_\alpha$ and $\eta_a$, one finds
\[
\partial_t \theta_{a,\alpha}+(c_a-U)\partial_x\theta_{a,\alpha}=0.
\]
Thus the internal composition of each degenerate block is simply
transported with speed $c_a-U$.

In particular, the loss of strict hyperbolicity has a simple meaning:
the genuinely nonlinear part of the dynamics is the reduced system for the
block totals $\eta_a$, while each block of multiplicity $\nu_a\ge2$
contributes $\nu_a-1$ linearly degenerate contact fields with speed
$c_a-U$. Accordingly, the Riemann problem in the degenerate case consists
of the usual shock/rarefaction waves for the reduced system, together with
contact discontinuities carrying the internal composition of the repeated
species.

For example, if $n=2$ and $f_1=f_2=c$, then
\[
\eta:=\rho_1+\rho_2
\]
satisfies
\[
\partial_t \eta+\partial_x\bigl(c\eta(1-\eta)\bigr)=0,
\]
while the fraction
\[
\theta:=\frac{\rho_1}{\rho_1+\rho_2}
\]
satisfies
\[
\partial_t \theta+c(1-\eta)\partial_x\theta=0.
\]
So one genuinely nonlinear family remains, and the second family becomes
a contact mode.

\section{PEP($n$) with open boundaries}
\label{open}
We now consider the process defined in the box $\mathbb{B}_L:=\{1,\dots L\}$ without any implied periodicity
on the site index $k$ of the occupation variables $\alpha_k$.
For $1\leq k \leq L-1$ particle exchanges between neighbouring sites
occur as defined above. However, at the boundary sites $1$ and
$L$
we describe particle exchange with hypothetical reservoirs
by the transitions $\beta \to \alpha$ with rates $b_{\alpha\beta}^{\pm}$
with superscripts $-$ and $+$ indicating the left and right boundary 
respectively. In analogy to the bulk exchange rates we also define
$b_{\alpha\alpha}:=0$ for $\alpha\in\{0,\dots,n\}$.\footnote{
Notice that the double subscript in the boundary rates
$b_{\alpha\beta}^{\pm}$ indicates the transition from state $\beta$
to the state $\alpha$ on the {\it same} boundary site 1 or $L$,
unlike the double subscript in the bulk rates $g_{\alpha\beta}$
which indicates the particle exchange $(\alpha,\beta) \to (\beta,\alpha)$
on the two neighbouring sites $(k,k+1)$.}

With the boundary generators $\mathcal{L}^{-}_{1}$ and $\mathcal{L}^{+}_{L}$
defined by
\begin{eqnarray}
(\mathcal{L}^{-}_1 f) (\balpha) & = &
 \sum_{\alpha=0}^n  \sum_{\beta=0}^n b_{\alpha,\beta}^{-} n^{\beta}_{1}
[ f(\alpha,\alpha_2,\dots,\alpha_L) - f(\balpha) ]
\label{genPEPm} \\
(\mathcal{L}^{+}_{L} f) (\balpha) & = &
 \sum_{\alpha=0}^n  \sum_{\beta=0}^n b_{\alpha,\beta}^{+} n^{\beta}_{L}
[ f(\alpha_1,\dots,\alpha_{L-1},\alpha) - f(\balpha) ]
\label{genPEPp}
 \end{eqnarray}
the generator of the PEP with open boundaries is given by
\begin{equation}
\mathcal{L}^{open}  = \mathcal{L}^{-}_1 
+ \sum_{k=1}^{L-1} \mathcal{L}^{ex}_{k,k+1} + \mathcal{L}^{+}_{L} .
\end{equation}

For the action of the generator on the occupation numbers $n^\alpha_k$ one then gets for $2\leq k \leq L-1$ 
the discrete continuity equation 
\begin{equation}
\mathcal{L}^{open} n^\alpha_{k} = j_{k-1}^\alpha - j_{k}^\alpha,
\label{ceopen}
\end{equation}
with the instantaneous bulk currents \eqref{instcurr}
and the boundary equations
\begin{equation}
\mathcal{L}^{open} n^\alpha_{1} = j_{-}^\alpha - j_1^\alpha, \quad
\mathcal{L}^{open} n^\alpha_{L} = j_{L-1}^\alpha - j_{+}^\alpha
\label{cepm}
\end{equation}
with the instantaneous boundary currents
\begin{equation}
j_{-}^{\alpha} =
\sum_{\beta=0}^n (b_{\alpha\beta}^{-}n^{\beta}_{1} - b_{\beta\alpha}^{-}n^{\alpha}_{1}) , \quad
j_{+}^{\alpha} =
\sum_{\beta=0}^n (b_{\beta\alpha}^{+}n^{\alpha}_{L} - b_{\alpha\beta}^{+}n^{\beta}_{L}) 
\label{instcurboundary}
\end{equation}

\subsection{Invariant measure}

The open PEP does not conserve particle number, however, for the invariant
measure the discrete continuity equation implies the current equalities
\begin{equation}
j^\alpha := \exval{j^\alpha_{-}} = \exval{j^\alpha_{k}} = \exval{j^\alpha_{+}}
\label{statcurrrel}
\end{equation}
for all $k\in\{1,\dots,L-1\}$. An invariant measure can be obtained in explicit
product form for some specific conditions on the boundary rates.
%
%
%

\begin{theo}
The product measure \eqref{invmeas} is invariant for the open PEP 
if and only if the boundary parameters are on the parameter
manifold 
\begin{equation}
\ba{l}
\displaystyle \sum_{\alpha=0}^n 
\left(b_{\beta,\alpha}^{-} \frac{\rho_\alpha}{\rho_\beta}  - b_{\alpha,\beta}^{-} 
+ b_{\alpha,0}^{-} 
- b_{0,\alpha}^{-} \frac{\rho_\alpha}{\rho_0}
\right)  - f_\beta = 0 \\[6mm]
\displaystyle \sum_{\alpha=0}^n 
\left(b_{\beta,\alpha}^{+} \frac{\rho_\alpha}{\rho_\beta}  - b_{\alpha,\beta}^{+} 
+ b_{\alpha,0}^{+} 
- b_{0,\alpha}^{+} \frac{\rho_\alpha}{\rho_0}
\right)  + f_\beta = 0 \\[6mm]
\displaystyle \sum_{\alpha=0}^n \left(b_{\alpha,0}^{-}  - 
b_{0,\alpha}^{-} \frac{\rho_\alpha}{\rho_0} \right)  
+ \sum_{\alpha=0}^n \left(b_{\alpha,0}^{+} - b_{0,\alpha}^{+} 
\frac{\rho_\alpha}{\rho_0}  \right)  = 0
\ea
\label{boundarymanifold}
\end{equation}
for all $\beta\in\{1,\dots,n\}$.
\end{theo}

\Proof 
We first prove that the conditions \eqref{boundarymanifold} are sufficient
for invariance of the product measure \eqref{invmeas}.

Since we are on a finite state space, invariance of a probability measure
$\pi$ under the action of a generator $\mathcal{L}$ can be proved by 
verifying  $(\mathcal{L}^\ast \pi) (\balpha) = 0$ with the adjoint generator \cite{Ligg10}, or, equivalently, 
\begin{equation}
(\pi^{-1} \mathcal{L}^\ast \pi) (\balpha) = 0 \quad \forall \balpha \in \Omega
\label{invprop}
\end{equation}
since the process is ergodic and therefore $\pi(\balpha) \neq 0$
for all $\balpha\in\Omega$.
The adjoint $\mathcal{L}^{open\ast}$ is obtained from
$\mathcal{L}^{open}$ by interchanging the indices $\alpha\leftrightarrow\beta$ in the transition rates. \\

Straightforward
computation yields due to the factorization
property of the measure the boundary relations
\begin{eqnarray}
(\pi^{-1} \mathcal{L}^{-\ast}_{1} \pi) (\balpha) 
& = & \sum_{\beta=1}^n A^{-}_\beta n^{\beta}_{1} + B^{-} \\
(\pi^{-1} \mathcal{L}^{+\ast}_{L} \pi) (\balpha) 
& = & \sum_{\beta=1}^n A^{+}_\beta n^{\beta}_{L} + B^{+}
\end{eqnarray}
with the constants
\begin{eqnarray}
A_\beta^{\pm} & = & \sum_{\alpha=0}^n 
\left(b_{\beta,\alpha}^{\pm} \frac{\rho_\alpha}{\rho_\beta}  - b_{\alpha,\beta}^{\pm} 
+ b_{\alpha,0}^{\pm} 
- b_{0,\alpha}^{\pm} \frac{\rho_\alpha}{\rho_0}
\right) \\
B^{\pm} & = & \sum_{\alpha=0}^n \left(b_{0,\alpha}^{\pm} 
\frac{\rho_\alpha}{\rho_0} - b_{\alpha,0}^{\pm} \right) 
\end{eqnarray}

The symmetry $w_{\alpha\beta} = w_{\beta\alpha}$,
and the antisymmetry  $f_{\alpha\beta} = - f_{\beta\alpha}$ lead to the
bulk relation
\begin{equation}
(\pi^{-1} \mathcal{L}^{ex\ast}_{k,k+1} \pi) (\balpha) =
\sum_{\beta=1}^n f_\beta (n^{\beta}_{k+1} - n^{\beta}_{k} ),
\quad 1 \leq k \leq L-1 
\end{equation}
and the therefore by the telescopic property of the sum
\begin{equation}
(\pi^{-1} \mathcal{L}^{ex\ast} \pi) (\balpha) =
\sum_{\beta=1}^n f_\beta (n^{\beta}_{L} - n^{\beta}_{1} )
\end{equation}
so that
\begin{equation}
(\pi^{-1} \mathcal{L}^{open\ast} \pi) (\balpha) 
= \sum_{\beta=1}^n (A^{-}_\beta - f_\beta) n^{\beta}_{1} 
+ \sum_{\beta=1}^n (A^{+}_\beta + f_\beta) n^{\beta}_{L} 
+ B^{-} + B^{+}
\end{equation}
for all $\balpha \in \Omega$. Since by assumption
the boundary parameters  are on the manifold \eqref{boundarymanifold},
the r.h.s. of this equation vanishes. This implies by \eqref{invprop} that the product measure \eqref{invmeas} is invariant.

Assume now that the product measure \eqref{invmeas} is invariant.
From \eqref{instcurboundary} one obtains
\begin{eqnarray*}
\exval{j^\alpha_{-}} 
 & = & 
\sum_{\beta=0}^n (b_{\alpha\beta}^{-} \rho_{\beta} - b_{\beta\alpha}^{-} \rho_{\alpha}) \\
\exval{j^\alpha_{+}} 
 & = & -
\sum_{\beta=0}^n 
(b_{\alpha\beta}^{+} \rho_{\beta} - b_{\beta\alpha}^{+} \rho_{\alpha}).
\end{eqnarray*}
With the current-density
relation \eqref{statcurrrel} the manifold \eqref{boundarymanifold} 
can thus be written as
\begin{equation}
\exval{j^\alpha_-} = \exval{j^\alpha_+} = j^\alpha
\end{equation}
for all $\alpha \in \{1,\dots,n\}$. By the discrete continuity equation \eqref{cepm} this is
a necessary condition for invariance of the product measure \eqref{invmeas}.
\qed

For boundary rates as given
above we say that the system is in contact with reservoirs of
densities $\rho_{\alpha}$.
This manifold is called {\it pde-friendly} \cite{Popk04} since 
there is no microscopic density boundary layer, i.e., 
$\exval{n^\alpha_k} = \rho_\alpha$ for all $k\in\{1,\dots,L\}$.

\subsection{Phase diagram of PEP($n$) with open boundaries}

The manifold \eqref{boundarymanifold} where the product measure \eqref{invmeas} is invariant involves the densities $\rho_\alpha$ that
define the measure in such a way that the first constraint in 
\eqref{boundarymanifold} involves only the left boundary rates
while the second constraint involves only the right boundary rates.
Thus one can invert these relations to obtain boundary densities $\rho_\alpha^{-}$
as functions of the boundary rates $b^{-}_{\alpha\beta}$ and 
{\it a priori} different densities $\rho_\alpha^{+}$
as functions of the boundary rates $b^{+}_{\alpha\beta}$.
The third constraint in \eqref{boundarymanifold} imposes the equality $\rho_\alpha^{+}=\rho_\alpha^{-}$ for all $\alpha\in\{1,\dots,n\}$ of the
densities for which a semi-infinite system with only one set of boundary rates would be stationary, given by either $\rho_\alpha^{+}$ or $\rho_\alpha^{-}$.
This means that in a finite system the left and right boundary densities
need to be equal to ensure that product measure \eqref{invmeas} is 
invariant.

Choosing $\rho_\alpha^{+}\neq \rho_\alpha^{-}$ for one or more
species $\alpha$ then defines an open PEP which still has a unique
invariant measure due to ergodicity, but this (unknown) invariant measure 
is not product and, in general, not expected to have a simple explicit structure, even though in some specific two-species particle-exchange models it can be expressed as a matrix product measure, see \cite{Alca98,Isae01,Arit06,Uchi08,Ayye09,Ayye12,Cram15,Cram16a}
for examples. Generally, for finite systems of size $L$ one expects the local density $\exval{n^\alpha_{k}}$ to depend on the coordinate $k$. As in one-species systems \cite{Derr93a,Schu93b,Popk99}
the densities of a multi-species system are expected to approach a bulk value
$\rho_\alpha
:= \lim_{L\to\infty} \exval{n^\alpha_{
\lfloor
x L\rfloor}}$ in the thermodynamic limit $L\to\infty$ that is independent of $x\in(0,1)$ but which depends on
the two sets of boundary densities
$\rho_\alpha^{+}$ and $\rho_\alpha^{-}$ in a way that is discontinuous
at manifolds of first-order boundary induced phase transitions and
continuous (but with discontinuous derivatives w.r.t. at least one of the densities) at manifolds of second-order boundary induced phase transitions
\cite{Popk04a,Popk11,Cant24}.

It is the purpose of this section to derive for the PEP($n$)
the bulk densities $\rho_\alpha$
of this unknown invariant measure as a function of the two sets of 
the microscopic boundary densities
$\rho_\alpha^{+}$ and $\rho_\alpha^{-}$ and hence, via the first two
constraints in
\eqref{boundarymanifold} as a function of the
boundary rates that define the process on the microscopic scale. To this end, we follow
the approach of \cite{Cant24,zahra2025hydrodynamics} relying on the hydrodynamic limit
for an infinite system studied in the previous section and then determine the steady-state selection
in terms of the Riemann invariants of the hydrodynamic equation.

Due to the pde-friendliness of the boundary rates
this result yields the bulk densities of the
true invariant measure as an explicit function of the {\it microscopic}
boundary rates
rather than as functions only of the {\it macroscopic} boundary densities
of the stationary solution of the hydrodynamic equation
which in the generic approach of \cite{Cant24} are unknown functions of the microscopic boundary rates. Thus, the stationary phase diagram of the
PEP($n$) is fully and explicitly given by the parameters that define the
microscopic model.

Denote the left and right boundary Riemann variables as:
\[
{\bf z}^-=(z_1^-,\dots,z_n^-),\qquad
{\bf z}^+=(z_1^+,\dots,z_n^+)
\]
Following the principle of \cite{Cant24}, we assume that
the bulk state of the open system is selected by the boundary Riemann
problem, namely
\[
{\bf z}^B={\bf z}(0),
\]
where ${\bf z}(\xi)$ is the self-similar entropy solution of the Riemann problem
with left state ${\bf z}^-$ and right state ${\bf z}^+$.

Throughout this section we work in the nondegenerate case
\[
f_0<f_1<\dots<f_n,
\qquad
\mathcal D_z=\prod_{k=1}^n (f_{k-1},f_k),
\]
so that
\[
{\bf z}^-,{\bf z}^+\in\mathcal D_z.
\]

\subsection{Steady state selection}

Recall that the Riemann solution is obtained by concatenating the $n$
elementary waves in increasing family order. Along the $k$-th wave, the
variables $z_1,\dots,z_{k-1}$ have already taken their right values,
whereas $z_{k+1},\dots,z_n$ still equal their left values. Hence the
$k$-th wave is governed by the scalar conservation law
\[
\partial_t q+\partial_x\bigl(q^2+C_k q\bigr)=0,
\qquad q=z_k,
\]
with
\[
C_k:=\sum_{j<k} z_j^+ + \sum_{j>k} z_j^- - \sum_{\alpha=0}^n f_\alpha.
\]
Equivalently, the characteristic speed along the $k$-th family is
\[
g_k'(q)=2q+C_k.
\]
This yields the following explicit selection rule for the bulk state.

Let ${\bf z}^-,{\bf z}^+\in\mathcal D_z$, and let \({\bf z}^B={\bf z}(0)\) be the value at $\xi=0$ of
the self-similar entropy solution of the Riemann problem with left state
${\bf z}^-$ and right state ${\bf z}^+$. Then, for each $k=1,\dots,n$,
\begin{equation}
z_k^B =
\begin{cases}
\min\!\left\{ z_k^+,\, \max\!\left\{ z_k^-,\, -\dfrac{C_k}{2} \right\} \right\},
& z_k^- < z_k^+, \\[3mm]
z_k^-,
& z_k^- > z_k^+ \ \text{and}\ z_k^- + z_k^+ + C_k > 0, \\[1mm]
z_k^+,
& z_k^- > z_k^+ \ \text{and}\ z_k^- + z_k^+ + C_k < 0.
\label{thm:open-bulk-selection}
\end{cases}
\end{equation}
If
\[
z_k^->z_k^+,
\qquad
z_k^-+z_k^+ + C_k=0,
\]
then the $k$-th shock is stationary, and one lies on a coexistence
hypersurface.

%
To derive this steady state selection principle, notice that
the $k$-th elementary wave is the scalar entropy solution of
\[
\partial_t q+\partial_x\bigl(q^2+C_k q\bigr)=0,
\qquad q=z_k,
\]
with left state $q=z_k^-$ and right state $q=z_k^+$. Its characteristic
speed is
\[
g_k'(q)=2q+C_k.
\]

If $z_k^-<z_k^+$, the solution is a rarefaction. Evaluating the
self-similar profile at $\xi=0$ gives
\[
z_k^B=
\begin{cases}
z_k^-, & -\dfrac{C_k}{2}<z_k^-,\\[2mm]
-\dfrac{C_k}{2}, & z_k^-\le -\dfrac{C_k}{2}\le z_k^+,\\[3mm]
z_k^+, & -\dfrac{C_k}{2}>z_k^+,
\end{cases}
\]
which is exactly
\[
z_k^B=\min\!\left\{ z_k^+,\, \max\!\left\{ z_k^-,\, -\dfrac{C_k}{2} \right\} \right\}.
\]

If $z_k^->z_k^+$, the solution is a shock with Rankine--Hugoniot speed
\[
s_k=\frac{(z_k^-)^2+C_k z_k^--\bigl((z_k^+)^2+C_k z_k^+\bigr)}{z_k^- - z_k^+}
=z_k^-+z_k^+ + C_k.
\]
Hence
\[
z_k^B=
\begin{cases}
z_k^-, & s_k>0,\\[1mm]
z_k^+, & s_k<0.
\end{cases}
\]
If $s_k=0$, the shock is stationary, yielding the coexistence
hypersurface.

The steady state selection principle \eqref{thm:open-bulk-selection} shows that each Riemann variable
$z_k$ can belong to one of the three phases:
\begin{itemize}
    \item \textbf{Left-Induced Phase  $\rm L^{(k)}$:} If \( \lambda_k({\bf z^B}) > 0 \), the value of \( z_k \) is determined by the left boundary, i.e., \( z_k^B = z_k^{-} \). This is the analogue of the low-density phase in single-species systems with
concave current-density relation.

    \item \textbf{Right-Induced Phase $ \rm R^{(k)}$: } If \( \lambda_k({\bf z^B})  < 0 \), the value of \( z_k \) is determined by the right boundary, i.e., \( z_k^B = z_k^{+} \). 
 This is the analogue of the high-density phase in single-species systems.
    \item \textbf{Bulk-Induced Phase $ \rm B^{(k)}$:} If \( \lambda_k({\bf z^B})  = 0 \), the value of \( z_k \) is not influenced by either boundary but instead lies on the hypersurface defined by \( \lambda_k({\bf z^B}) = 0 \). In other words, this happens if the $k$-th wave is a rarefaction
crossing the origin, so that
$z_k^B=-\frac{C_k}{2}$.
This is the analogue of the maximal-current phase in single-species systems.
\end{itemize}

\subsubsection*{Global phase diagram}

Because the elementary waves are strictly ordered, at most one family can
contribute a rarefaction that expands on both sides of the origin. Therefore at most one
component of~${\bf z^B}$ can be bulk-induced. This yields $2n+1$ generic phases:
$n+1$ {\it boundary-induced} phases, in which every component is selected
from one of the two reservoirs, and $n$ {\it mixed} phases, in which
exactly one component is bulk-induced.

We denote by~$\mathcal{P}_q$, for $0\le q\le n$, the boundary-induced
phase in which the first~$q$ components are induced from the right and the
remaining~$n-q$ from the left:
\[
\mathcal{P}_q
:=
\bigl(\mathrm{R}^{(1)},\dots,\mathrm{R}^{(q)},\,
      \mathrm{L}^{(q+1)},\dots,\mathrm{L}^{(n)}\bigr),
\qquad 0\le q\le n.
\]
In particular, $\mathcal{P}_0$ is the phase in which all components are
induced from the left, and $\mathcal{P}_n$ the one in which all are
induced from the right.

We denote by~$\mathcal{M}_p$, for $1\le p\le n$, the mixed phase in which
the first~$p-1$ components are induced from the right, the $p$-th
component is bulk-induced, and the remaining~$n-p$ are induced from the
left:
\[
\mathcal{M}_p
:=
\bigl(\mathrm{R}^{(1)},\dots,\mathrm{R}^{(p-1)},\,
      \mathrm{B}^{(p)},\,
      \mathrm{L}^{(p+1)},\dots,\mathrm{L}^{(n)}\bigr),
\qquad 1\le p\le n.
\]

The $2n+1$ phases are naturally ordered as
\[
\mathcal{P}_0,\quad
\mathcal{M}_1,\quad
\mathcal{P}_1,\quad
\dots,\quad
\mathcal{M}_n,\quad
\mathcal{P}_n,
\]
reflecting the left-to-right ordering of the wave families: as one moves
through parameter space, successive waves cross the origin one at a time.

The mixed phase~$\mathcal{M}_p$ is sandwiched between~$\mathcal{P}_{p-1}$
and~$\mathcal{P}_p$. The transition
$\mathcal{P}_{p-1}\leftrightarrow\mathcal{M}_p$ occurs when the left edge
of the $p$-th rarefaction crosses the origin ($2z_p^-+C_p=0$), while
$\mathcal{M}_p\leftrightarrow\mathcal{P}_p$ occurs when the right edge
does ($2z_p^++C_p=0$). When $z_p^->z_p^+$, the $p$-th wave is a shock and
the mixed phase~$\mathcal{M}_p$ is absent; there is instead a direct
coexistence transition $\mathcal{P}_{p-1}\leftrightarrow\mathcal{P}_p$ at
$z_p^-+z_p^++C_p=0$.

Away from these hypersurfaces, the space
$\mathcal{D}\times\mathcal{D}$ is partitioned into the $2n+1$ generic
phases listed above.

The phase structure becomes particularly transparent when viewed
family by family. Once the Riemann variables of all other families
are fixed, the constant $C_k$ is determined, and the selection rule~\ref{thm:open-bulk-selection}
expresses $z_k^B$ as a function of only the two boundary values
$(z_k^-, z_k^+)$. This reduces the $k$-th family to an effective
single-species problem whose phase diagram is shown in
Figure~\ref{fig:diagzk}. The dependence on the other
families enters only through the location of the triple point
$(-C_k/2, -C_k/2)$, around which the three phase regions
$\mathrm{L}^{(k)}$, $\mathrm{B}^{(k)}$, $\mathrm{R}^{(k)}$ are
organized. The full phase diagram of the $n$-species system is
obtained by superimposing these $n$ scalar pictures, one for each
family.

\begin{figure}[h!]
    \centering
  \begin{tikzpicture}[scale=1.2]
  \def\a{1}      
  \def\b{6}      
  \def\c{3}      
  \draw[very thick, fill=white!95!blue] (\a,\a) rectangle (\b,\b);
  \draw[thick, ->] (-0.3,0) -- (6.8,0) node[below] {$z_k^{-}$};
  \draw[thick, ->] (0,-0.3) -- (0,6.8) node[left]  {$z_k^{+}$};
  \node[below left] at (0,0) {$0$};
  \draw (\a,0.08) -- (\a,-0.08) node[below] {$f_{k-1}$};
  \draw (\b,0.08) -- (\b,-0.08) node[below] {$f_{k}$};
  \draw (0.08,\a) -- (-0.08,\a) node[left]  {$f_{k-1}$};
  \draw (0.08,\b) -- (-0.08,\b) node[left]  {$f_{k}$};
  \draw[dotted] (\a,0) -- (\a,\a);
  \draw[dotted] (0,\a) -- (\a,\a);
  \draw[dotted] (\b,0) -- (\b,\b);
  \draw[dotted] (0,\b) -- (\b,\b);
  \draw[dotted] (\c,0) -- (\c,\c);
  \draw[dotted] (0,\c) -- (\c,\c);
  \draw[dashed, thick] (\a,\a) -- (\b,\b);
  \draw[line width=2pt] (\c,\c) -- (\c,\b);
  \draw[line width=2pt] (\a,\c) -- (\c,\c);
  \pgfmathsetmacro{\xend}{min(\b, 2*\c-\a)}
  \pgfmathsetmacro{\yend}{max(\a, 2*\c-\b)}
  \draw[line width=2.2pt, brown] (\c,\c) -- (\xend,\yend);
  \fill (\c,\c) circle (1.5pt);
  \draw (\c,0.08) -- (\c,-0.08) node[below] {$-\tfrac{C_k}{2}$};
  \draw (0.08,\c) -- (-0.08,\c) node[left]  {$-\tfrac{C_k}{2}$};
  \node at (4,5.2) {$\mathrm{L}^{(k)}$};   
  \node at (1.9,4.5) {$\mathrm{B}^{(k)}$};   
  \node at (1.8,2.5) {$\mathrm{R}^{(k)}$};   
  \node at (4.7,3.2) {$\mathrm{L}^{(k)}$};   
  \node at (3,2) {$\mathrm{R}^{(k)}$};   
\end{tikzpicture}
    \caption{Phase diagram for $z_k^B$ in the $(z_k^-, z_k^+)$-plane
at fixed $C_k$, with domain $(f_{k-1}, f_k)^2$. The dashed diagonal separates rarefactions (above) from shocks (below). Solid black lines are the second-order transitions bounding the bulk-induced phase $\mathrm{B}^{(k)}$, and the brown line is the first-order coexistence line where the $k$-th shock is stationary. The dependence on the other Riemann variables enters only through the position of the triple point $(-C_k/2,\,-C_k/2)$ where the three phase boundaries meet.} 
    \label{fig:diagzk}
\end{figure}

As an example, consider $n=1$: the system reduces to a single-species
exclusion process, and the three phases
$\mathcal{P}_0$, $\mathcal{M}_1$, $\mathcal{P}_1$
correspond to the familiar low-density, maximal-current, and high-density
phases, respectively. The first nontrivial case is the two-component
system $n=2$, which we address in the next section.

\subsection{Example: 2-PEP with open boundaries}

The case $n=2$ is the first genuinely multicomponent example. The
corresponding hydrodynamic system can be mapped to a Leroux system \cite{Toth03} 
where the hydrodynamic limit before the emergence of the shocks was proven rigorously. Here we present it in the PEP
parametrization.
To simplify the notations, we fix
\[
 f_0=0,\qquad f_1=\alpha,\qquad f_2=1,
 \qquad 0<\alpha<1.
\]
This is the natural two-species representative of the nondegenerate
regime treated above. The Riemann variables satisfy
\[
0<z_1<\alpha<z_2<1,
\]
and the inverse map is explicitly given by
\[
\rho_0=\frac{z_1z_2}{\alpha},\qquad
\rho_1=\frac{(\alpha-z_1)(z_2-\alpha)}{\alpha(1-\alpha)},\qquad
\rho_2=\frac{(1-z_1)(1-z_2)}{1-\alpha}.
\]
Moreover, the characteristic speeds read
\[
\lambda_1(z_1,z_2)=2z_1+z_2-(1+\alpha),
\qquad
\lambda_2(z_1,z_2)=z_1+2z_2-(1+\alpha),
\]
so that
\[
\lambda_1(z_1,z_2)<\lambda_2(z_1,z_2)
\qquad\text{for all }(z_1,z_2)\in(0,\alpha)\times(\alpha,1).
\]
Hence only five of the nine a priori possible combinations of
left-boundary-selection, right-boundary-selection and bulk-selection can
occur. These are listed in Table~\ref{tab:pep2-phases}.

\begin{table}[h!]
\centering
\begin{tabular}{c||c|c|c}
 & $\lambda_1<0$ & $\lambda_1=0$ & $\lambda_1>0$ \\\hline\hline
$\lambda_2<0$ & $\mathcal P_2$ & $\times$ & $\times$ \\\hline
$\lambda_2=0$ & $\mathcal M_2$ & $\times$ & $\times$ \\\hline
$\lambda_2>0$ & $\mathcal P_1$ & $\mathcal M_1$ & $\mathcal P_0$
\end{tabular}
\caption{Admissible stationary phases of the open two-species PEP. The
strict ordering $\lambda_1<\lambda_2$ excludes the four missing sign
combinations.}
\label{tab:pep2-phases}
\end{table}

The bulk phase portrait is particularly transparent in the $(z_1,z_2)$
variables. The curves
\[
\lambda_1(z_1,z_2)=0
\qquad\text{and}\qquad
\lambda_2(z_1,z_2)=0
\]
separate the three boundary-selected regions $\mathcal P_0$,
$\mathcal P_1$, $\mathcal P_2$, while the mixed phases
$\mathcal M_1$ and $\mathcal M_2$ reduce to these two curves. In the
physical density domain, obtained through the inverse map above, the same
structure is represented by two smooth curves; see
Figure~\ref{fig:pep2-bulk-phases}. In this representation the mixed
phases are one-dimensional sets, not two-dimensional regions.

\begin{figure}[h!]
    \centering
\begin{tikzpicture}[scale = 0.8]
\draw [very thick]
    	[fill=white!95!blue] (0,0)--(4,0) node [below] {$\alpha$}--(4,6)--(0,6)node [left] {$1$} -- cycle;
			\draw[thick, ->] (-0.7,0) -- (5,0) node[below]
   {$z_{1}$};
			\draw[thick, ->] (0,-0.7) -- (0,7) node[left] {$z_{2}$};
   \draw (0.5,0) node[below left] {$0$};
   \draw (0,0.2) node[left] {$\alpha$};

   \draw[very thick, blue] (2,6) -- (2.7,4.6);
   \draw (3,4) node {\color{blue}$\mathcal{M}_1$};
   \draw[very thick, blue] (3.3,3.4) -- (4,2);

   \draw[very thick, red] (0,3) -- (1.6,2.2);
   \draw (2,2) node {\color{red}$\mathcal{M}_2$};
   \draw[very thick, red] (2.4,1.8) -- (4,1);

   \draw (3.3,5.2) node {$\mathcal{P}_0$};
   \draw (1.2,3.2) node {$\mathcal{P}_1$};
   \draw (1.5,0.7) node {$\mathcal{P}_2$};
\end{tikzpicture}
\begin{tikzpicture}[scale = 0.8]
\draw [very thick]
    	[fill=white!95!blue] (0,0)--(6,0) node [below] {$1$}--(0,6)node [left] {$1$} -- cycle;
			\draw[thick, ->] (-0.7,0) -- (7,0) node[below]
   {$\rho_{1}$};
			\draw[thick, ->] (0,-0.7) -- (0,7) node[left] {$\rho_{2}$};
   \draw (0.5,0) node[below left] {$0$};
   \draw (0,0.3) node[left] {$0$};

   \draw (0.6,4.2) node {$\mathcal{P}_2$};
   \draw (2.7,1.8) node {$\mathcal{P}_1$};
   \draw (0.7,0.6) node {$\mathcal{P}_0$};

  \draw[very thick, blue, scale = 6]
    plot[domain=0.2:0.4, samples=300, smooth, variable=\t]
    ({5/3 - 7.5*\t + (25/3)*\t*\t},
     {-2/3 + 4*\t - (10/3)*\t*\t});
  \draw (2.2,1.0) node {\color{blue}$\mathcal{M}_1$};

  \draw[very thick, red, scale = 6]
    plot[domain=0:0.4, samples=400, smooth, variable=\t]
    ({0.5 - (25/12)*\t + (25/12)*\t*\t},
     {0.5 + (1/3)*\t - (5/6)*\t*\t});
  \draw (1.5,3.5) node {\color{red}$\mathcal{M}_2$};
\end{tikzpicture}
\caption{Bulk phase portrait for the open two-species PEP, illustrated
for $\alpha=0.4$. Left: the physical rectangle
$0<z_1<\alpha<z_2<1$ in Riemann variables. Right: the corresponding
physical density triangle. The three boundary-selected phases occupy open
regions, while the mixed phases collapse to the curves
$\lambda_1=0$ and $\lambda_2=0$.}
\label{fig:pep2-bulk-phases}
\end{figure}

The boundary phase diagram is obtained by evaluating the selection rule of
Theorem~\ref{thm:open-bulk-selection} family by family. For $n=2$ one has
\[
C_1=z_2^--(1+\alpha),
\qquad
C_2=z_1^+-(1+\alpha).
\]
Hence the five phases are
\[
\mathcal P_0:\quad (z_1^B,z_2^B)=(z_1^-,z_2^-),
\]
\[
\mathcal M_1:\quad (z_1^B,z_2^B)=\left(\frac{1+\alpha-z_2^-}{2},\,z_2^-\right),
\]
\[
\mathcal P_1:\quad (z_1^B,z_2^B)=(z_1^+,z_2^-),
\]
\[
\mathcal M_2:\quad (z_1^B,z_2^B)=\left(z_1^+,\,\frac{1+\alpha-z_1^+}{2}\right),
\]
\[
\mathcal P_2:\quad (z_1^B,z_2^B)=(z_1^+,z_2^+).
\]

This can be represented by two planar slices of the four-dimensional
boundary space. In the $(z_1^-,z_1^+)$-plane, for fixed $z_2^-$, the
second-order transition lines are
\[
z_1^- = \frac{1+\alpha-z_2^-}{2},
\qquad
z_1^+ = \frac{1+\alpha-z_2^-}{2},
\]
while the coexistence line is
\[
z_1^-+z_1^+ + z_2^- -(1+\alpha)=0.
\]
Similarly, in the $(z_2^-,z_2^+)$-plane, for fixed $z_1^+$, the
second-order transition lines are
\[
z_2^- = \frac{1+\alpha-z_1^+}{2},
\qquad
z_2^+ = \frac{1+\alpha-z_1^+}{2},
\]
while the coexistence line is
\[
z_2^-+z_2^+ + z_1^+ -(1+\alpha)=0.
\]

This example illustrates the general mechanism in the simplest genuinely
multicomponent setting. The strict ordering $\lambda_1<\lambda_2$ rules
out the phase combinations in which the second family would be induced
from the left while the first one is either bulk-selected or induced from
the right. As a consequence, the two-species PEP exhibits exactly the five
phases predicted by the general theory.

\section{Conclusions}
\label{sec:conclusion}
In this work, we have shown that the $n$-species particle exchange process
provides a tractable example of a genuinely nonlinear multicomponent
hydrodynamic theory. The product-measure structure of the microscopic
model leads on large scales to hydrodynamic equations, viz., coupled inviscid Burgers equations, whose characteristic structure
can be analyzed globally and explicitly in terms of Riemann invariants, for an arbitrary number of species. This is a
strong property: for generic systems of three or more conservation laws,
the existence of a complete set of Riemann invariants is exceptional, since
generically the corresponding compatibility conditions are overdetermined
\cite{Lax73,Lax00,Serr00,Serre}.
From this point of view, the PEP hydrodynamics may be regarded as a
natural multicomponent analogue of Burgers hydrodynamics. The scalar ASEP
case is recovered when there is only one conserved density
\cite{Gart88,Kipn99}, while for general $n$ the elementary waves retain a
Burgers-type character in the appropriate Riemann variables. Thus the
model combines two features that are rarely simultaneously available in
systems with several conservation laws: a microscopic exclusion process
with factorized invariant measures, and a macroscopic system whose Riemann
problem can be solved explicitly.

The same structure is particularly useful in the presence of open
boundaries. It allows one to implement, in a fully explicit setting, the
hydrodynamic selection mechanism developed for boundary-driven
multispecies systems \cite{Popk04,Cant24,zahra2025hydrodynamics}. In
contrast with the scalar case, where the stationary bulk density can
be characterized by the extremal-current principle
\cite{Popk99}, the multicomponent bulk state is selected by the
full characteristic structure of the hyperbolic system
 in terms of the macrosopic reservoir densities. Remarkably, for the PEP this
selection can be expressed directly
in terms of the microscopic boundary rates on the
pde-friendly manifold. It is also interesting to note that the scalar
extremal-current principle can thus be understood as a special case of the 
entropy solution
 of the associated Riemann problem. 

The resulting open-boundary phase diagram provides, to our knowledge, one
of the few explicit examples of a boundary-driven exclusion process with
an arbitrary number of components beyond the two-species case. Other
examples include \cite{Ayye17,Roy21}; however, the present construction
is fundamentally different in that it does not require knowledge of the invariant measure
for open boundaries which is generally difficult to obtain. Instead, it is 
derived directly from the hydrodynamic Riemann
problem and applies uniformly to the full $n$-component PEP family.

Several directions remain open. 
First, we did not address the question of universality of the density profile
in the bulk-induced phase which in generic one-component systems approaches the bulk value with a universal power law and a universal
amplitude \cite{Hage01}. This result follows from nonlinear fluctuating
hydrodynamics arguments which we expect to remain valid for the PEP($n$)
with non-degenerate drift parameters studied in the present work. It is unclear, however, what universal features of the density profile
to expect in generic open particle systems with more than one bulk conservation law.

Second,  the degenerate case in which several drift parameters $f_\alpha$
coincide deserves further investigation. We have shown that the
corresponding species combine into block densities and that the remaining
internal compositions are transported through contact modes. The two-species
case with open boundaries, i.e., when $f_1=0$ and $f_2\neq 0$
%
was discussed in detail in \cite{Rako04}. The time-evolution of the density of one species is autonomous and microscopically identical to that of the single-species ASEP which is governed on Eulerian scale by the inviscid Burgers equation. Thus
in this special case the stationary phase diagram of the open PEP(2) can be deduced directly
from the known stationary phase diagram of the single-species ASEP
obtained in \cite{Derr93a,Schu93b}.

When $f_1=f_2=0$ 
the PEP(2) is expected to be non-trivial under diffusive scaling
where it is governed by a 2-species diffusion equation. Specifically, for $w_{10}=w_{20}=w_{21}$ the PEP(2) is equivalent to the self-dual
two-species symmetric simple exclusion process with colored particles \cite{Schu94,Giar09,Sa25}
which does not exhibit boundary-induced phase transitions.
The stationary density profiles $\exval{n^1_k}$ and $\exval{n^2_k}$ 
are linear in $k$ and one finds universal properties of the correlations
\cite{Cari13,Casi24} similar to
those found in the single-species simple symmetric exclusion process 
\cite{Spoh83,Derr93a,Flor23}
and in the boundary-driven XXX quantum Heisenberg chain \cite{Pros11,Popk13,Buca16}. 

For $w_{10}\neq w_{20}$ and $w_{21}=0$
the process is equivalent to the symmetric simple exclusion process with
tagged particles \cite{Arra83,Quas92} which is also a model with experimental relevance for single-file diffusion   \cite{Kukl96,Wei00} and the
dynamics of entangled polymers 
\cite{Perk94,Bark96,Schu99}.
A boundary-induced phase transition was shown to occur \cite{Brza07}
by appealing to hydrodynamic arguments employed rigorously for the infinite
system \cite{Quas92}. This diffusive boundary-induced
phase transition derives microscopically from the subdiffusive single-file behaviour of the tagged particles \cite{Arra83} which can be seen as ``defect particles'' that block the motion of the first particle species. It is thus
of a different nature than those in driven systems and
not within the scope of the present work on steady-state selection in bulk-driven systems.

More generally, the
appearance of these contact waves suggests a relation with defect-like or
passive degrees of freedom inside the hydrodynamic theory. This viewpoint
is reminiscent of recent work on a single impurity in TASEP, where a
distinguished particle with modified microscopic rates can have a
nontrivial macroscopic effect on the surrounding density field
\cite{CantZah25}. Understanding whether the contact modes arising from
degenerate drift parameters can be interpreted in a similar way could
clarify the role of degeneracies in multispecies exclusion processes.

A related problem concerns the presence of umbilical points where one or
more characteristic velocities of the macroscopic system of conservation
laws coincide. An example is the two-lane model in \cite{Popk12} which exhibits unusual shock waves and which has recently again attracted attention in the study stationary density fluctuations \cite{Roy25,Spoh26}. In such a scenario the 
Riemann approach of the present work fails and the problem
of steady state selection remains open.


\paragraph{Acknowledgements}
This work was financially supported by the
ANR-PRME Uniopen project (ANR-22-CE30-0004-01), 
FCT/Portugal through project
UIDB/04459/2020
(DOI: \url{https://doi.org/10.54499/UIDP/04459/2020}),
as well as through grants 2020.03953.CEECIND
(DOI: \url{https://doi.org/10.54499/2020.03953.CEECIND/CP1587/CT0013})
and 2022.09232.PTDC
(DOI: \url{https://doi.org/10.54499/2022.09232.PTDC}).

\paragraph{Data availability statement}

This is a theoretical study and no new data were created or analysed in this work.

\paragraph{Conflict of Interest Statement}
The authors declare that they have no conflict of interest.



\begin{thebibliography}{99}

%




\bibitem{Alca98}
F. C. Alcaraz, S. Dasmahapatra, and V. Rittenberg,
N-species stochastic models with boundaries and quadratic algebras,
J. Phys. A: Math. Gen. \textbf{31,} 845--878 (1998).

\bibitem{Arit06} 
Arita, C. 
Phase transitions in the two-species totally asymmetric exclusion process with open boundaries. 
J. Stat. Mech. Theory Exp. P12008, (2006).

\bibitem{Arra83} R. Arratia, The motion of a tagged particle in the simple symmetric exclusion system in Z,Ann. Probab. \textbf{11}, 362--373 (1983)

\bibitem{Ayye09}  
Ayyer, A., Lebowitz, J. L. and Speer, E. R. 
On the two species asymmetric exclusion process with semi-permeable boundaries. 
J. Stat. Phys. 135, 1009–1037 (2009).

\bibitem{Ayye12}   
Ayyer, A., Lebowitz, J. L. and Speer, E. R. 
On some classes of open two-species exclusion processes. 
Markov Process. Related Fields 18, 157–176 (2012).

\bibitem{Ayye17}
A. Ayyer and D. Roy,
The exact phase diagram for a class of open multispecies asymmetric exclusion processes,
Scientific Reports \textbf{7},  13555 (2017) .

\bibitem{Baha12}
C. Bahadoran,
Hydrodynamics and Hydrostatics for a Class of 
Asymmetric Particle Systems with Open Boundaries,
Commun. Math. Phys. {\bf 310}, 1 -- 24 (2012).

%

\bibitem{Bark96} G.T. Barkema, C. Caron, and J.F. Marko,  Scaling properties of gel electrophoresis of DNA, Biopolymers \textbf{38}, 665--667 (1996).

\bibitem{Brza07}
A. Brzank and G.M. Sch\"utz,
Phase Transition in the two-component symmetric exclusion process 
with open boundaries, 
J. Stat. Mech. \textbf{2007}, P08028 (2007).

\bibitem{Buca16} Bu\v{c}a, B. and Prosen, T.: B. Bu\v{c}a and T. Prosen, Connected correlations, fluctuations and current of magnetization in the steady state of boundary driven XXZ spin chains. J. Stat. Mech. \textbf{2016}, 023102 (2016).


\bibitem{Cant16} 
Cantini, L., Garbali, A., de Gier, J. and Wheeler, M. 
Koornwinder polynomials and the stationary multi-species asymmetric exclusion process with open boundaries. 
J. Phys. A 49, 444002 (2016).

\bibitem{Cant22}
L. Cantini and A. Zahra, 
Hydrodynamic behavior of the two-TASEP,
J. Phys. A: Math. Theor. \textbf{55} 305201 (2022).

\bibitem{Cant24}
L. Cantini and A. Zahra,
Steady-state selection in multi-species driven diffusive systems, 
EPL \textbf{146}, 21006 (2024)

\bibitem{CantZah25}
L. Cantini and A. Zahra,
Single impurity in the totally asymmetric simple exclusion process,
J. Stat. Mech. \textbf{2025}, 043204 (2025).


\bibitem{Cari13} G. Carinci, C. Giardin\`a, C. Giberti, and F. Redig, Duality for Stochastic Models of Transport, J. Stat. Phys. \textbf{152}, 657--697 (2013).

\bibitem{Casi24} F. Casini, R. Frassek, and C. Giardin\`a, Duality for the multispecies stirring process with open boundaries, J. Phys. A: Math. Theor. \textbf{57}, 295001 (2024)

\bibitem{Chen95} 
G.-Q. Chen and C.M. Dafermos, 
The vanishing viscosity method in one-dimensional thermoelasticity, 
Trans. Amer. Math. Soc. \textbf{347}, 531–541 (1995)

\bibitem{Chow24}
D. Chowdhury, A. Schadschneider, K. Nishinari,
Physics of collective transport and traffic phenomena in biology: Progress in 20 years,
Physics of Life Reviews \textbf{51},
409--422 (2024).

\bibitem{Cram15}
Crampe, N., Mallick, K., Ragoucy, E. and Vanicat, M. 
Open two-species exclusion processes with integrable boundaries. 
J. Phys. A 48, 175002, (2015).

\bibitem{Cram16a} 
Crampe, N., Evans, M., Mallick, K., Ragoucy, E. and Vanicat, M. 
Matrix product solution to a 2-species tasep with open integrable boundaries. arXiv preprint arXiv:1606.08148 (2016).

\bibitem{Cram16b} 
Crampe, N., Finn, C., Ragoucy, E. and Vanicat, M. 
Integrable boundary conditions for multi-species asep. 
J. Phys. A 49, 375201, (2016).


\bibitem{Derr93a}
B.~Derrida, M.R.~Evans, V.~Hakim, and V.~Pasquier,
Exact solution of a 1D asymmetric exclusion model using a matrix formulation,
J. Phys. A: Math. Gen. \textbf{26}, 1493--1517 (1993).

\bibitem{Derr93b}
Derrida, B., Janowsky, S.A., Lebowitz, J.L., Speer, E.R.: 
Exact solution of the totally asymmetric simple exclusion process: Shock profiles, 
J. Stat. Phys. \textbf{73}, 813--842 (1993).

%


\bibitem{Dubrovin}
Dubrovin BA, Novikov SP. Hydrodynamics of weakly deformed soliton lattices. Differential geometry and Hamiltonian theory. Russian Mathematical Surveys. 1989 Dec 31;44(6):35-124.

\bibitem{Evan98b} 
M.R. Evans,  Y. Kafri, H.M. Koduvely, and D. Mukamel,
Phase separation and coarsening in one-dimensional driven diffusive 
systems: Local dynamics leading to long-range Hamiltonians,
Phys. Rev. E {\bf 58} 2764--2778 (1998).

\bibitem{Ferr94a} 
P.A. Ferrari and L.R.G. Fontes, 
Shock fluctuations in the asymmetric simple exclusion process.
Probab. Theory Relat. Fields \textbf{99}, 305--319 (1994).

\bibitem{Ferr13}
P.L. Ferrari, T. Sasamoto and H. Spohn,
Coupled Kardar-Parisi-Zhang equations in one dimension, 
J. Stat. Phys.  \textbf{153}, 377--399 (2013).

\bibitem{Flor23} S. Floreani, A.G. Casanova Non-equilibrium steady state of the symmetric exclusion process with reservoirs, arXiv:2307.02481 [math.PR] (2023).

\bibitem{Gart88}
G\"artner J.: 
Convergence towards Burgers’ equation and propagation of chaos for weakly asymmetric exclusion processes. 
Stoch. Proc. Appl. \textbf{27}, 233--260 (1988).

\bibitem{Giar09} C. Giardin\`a, J. Kurchan, F. Redig,  and K. Vafayi,
Giardin\`a, C., Kurchan, J., Redig, F., and Vafayi, K.: Duality and Hidden Symmetries in Interacting Particle Systems, J. Stat. Phys. \textbf{135}, 25--55 (2009).

\bibitem{Gris11}
R. Grisi and G.M. Sch\"utz,
Current symmetries for particle systems with several conservation laws,
J. Stat. Phys. \textbf{145}, 1499--1512 (2011).

\bibitem{GrosSpo03}
S. Gro{\ss}kinsky and H. Spohn,
Stationary measures and hydrodynamics of zero range processes with several
species of particles,
Bull. Braz. Math. Soc. \textbf{34}, 489--507 (2003).

\bibitem{Hage01} 
J.S. Hager J. Krug, V. Popkov, and G.M. Sch\"utz,
Minimal current phase and universal boundary layers in driven 
diffusive systems
Phys. Rev. E \textbf{63}, 056110 (2001).

\bibitem{Isae01}
A. P. Isaev, P. N. Pyatov and V. Rittenberg,
Diffusion algebras,
J. Phys. A: Math. Gen. \textbf{34}, 5815--5834 (2001).


\bibitem{Kipn99}
Kipnis, C., Landim, C.:
Scaling limits of interacting particle systems.
Springer, Berlin (1999).

\bibitem{Krug91}
J Krug, 
Boundary-induced phase transitions in driven diffusive systems. 
Phys. Rev. Lett. \textbf{67}, 1882--1885 (1991).

\bibitem{Kukl96} Kukla, V., Kornatowski, J., Demuth, D., Girnus, I., Pfeifer, H., Rees, L.V.C., Schunk, S., Unger, K. and K\"arger, J.:  NMR studies of single-file diffusion in unidimensional channel zeolites. Science \textbf{272} 702--704 (1996).

\bibitem{Lax73}
P. D. Lax, Hyperbolic Systems of Conservation Laws and
the Mathematical Theory of Shock Waves, Vol. 11 (SIAM,
Philadelphia, 1973).

\bibitem{Lax00}
P. D. Lax, Hyperbolic Partial Differential Equations, Courant
Lecture Notes in Mathematics, Vol. 14 (AMS/Courant Institute
of Mathematical Sciences, New York, 2000).


\bibitem{Ligg99} 
Liggett, T.M.:
Stochastic Interacting Systems: Contact, Voter and Exclusion Processes.
Springer, Berlin (1999).

\bibitem{Ligg10} 
T.M.~Liggett, 
Continuous Time Markov Processes: An Introduction,
Graduate Studies in Mathematics Vol. 113,
American Mathematical Society, Rhode Island (2010).

\bibitem{MacD68} 
MacDonald J.T., Gibbs J.H., and Pipkin A.C.:
Kinetics of biopolymerization on nucleic acid templates.
Biopolymers \textbf{6}, 1--25 (1968).


\bibitem{Perk94}  T.T. Perkins, D.E. Smith, and S. Chu, Direct observation of tube-like motion of a single polymer-chain, Science {\bf 264}, 819--822 (1994).

\bibitem{Popk99} 
Popkov V and Sch\"utz G M 
Steady-state selection in driven diffusive systems with open 
boundaries,
Europhys. Lett. \textbf{48}, 257--264 (1999)

\bibitem{Popk04a}
V. Popkov,
Infinite reflections of shock fronts in driven diffusive systems with two species
J. Phys. A: Math. Gen. \textbf{37}, 1545--1557 (2004).

\bibitem{Popk04} 
V. Popkov and G.M. Sch\"utz,
Why spontaneous symmetry breaking disappears in a bridge system with PDE-friendly boundaries,
 J. Stat. Mech. \textbf{2004}, P12004 (2004).

\bibitem{Popk11}
V. Popkov and M. Salerno, 
Hierarchy of boundary-driven phase transitions in multispecies particle systems,
Phys. Rev. E \textbf{83}, 011130 (2011).

\bibitem{Popk12}
V. Popkov and G. M. Schütz,
Unusual shock wave in two-species driven systems with an umbilic point,
Phys. Rev. E \textbf{86}, 031139 (2012)

\bibitem{Popk13} V. Popkov, D. Karevski, and G. M. Schütz,  Driven isotropic Heisenberg spin chain with arbitrary boundary twisting angle: exact results,  Phys. Rev. E \textbf{88}, 062118 (2013).

%
\bibitem{Pros11} T. Prosen, Exact Nonequilibrium Steady State of a Strongly Driven Open XXZ Chain. Phys. Rev. Lett. \textbf{107}, 137201 (2011).

\bibitem{Quas92} J. Quastel, Diffusion of color in the simple exclusion process, Commun. Pure Appl. Math. \textbf{45}, 623--679 (1992).


\bibitem{Rako04} 
A. R\'akos and G.M. Sch\"utz,
Exact shock measures and steady state selection in a driven
diffusive system with two conserved densities,
J. Stat. Phys. \textbf{117}, 55-76 (2004).


\bibitem{Reza91}
F. Rezakhanlou,
Hydrodynamic limit for attractive particle systems on $\mathbb Z^d$,
Commun. Math. Phys. \textbf{140}, 417--448 (1991).

\bibitem{Roy21}
D. Roy,
The phase diagram for a class of multispecies permissive asymmetric exclusion processes,
J. Stat. Mech. \textbf{2021}, 013201 (2021)

\bibitem{Roy25}
D. Roy, A. Dhar, M. Kulkarni and H. Spohn,
Fixed points and universality classes in coupled Kardar–Parisi–Zhang equations, 
J. Stat. Mech. \textbf{2025},  073209 (2025)

\bibitem{Sand94b}
S. Sandow, 
Partially asymmetric exclusion process with open boundaries,
Phys. Rev. E \textbf{50}, 2660--2667 (1994). 

\bibitem{Schu93b}
G. Sch\"{u}tz and E. Domany, 
Phase transitions in an exactly soluble one- dimensional asymmetric exclusion model. 
J. Stat. Phys. \textbf{72}, 277--296 (1993).

\bibitem{Schu23a} G.M. Sch\"utz,   A reverse duality for the ASEP with open boundaries,  J. Phys. A: Math. Theor. \textbf{56}, 274001 (2023).

\bibitem{Sa25} L. S\'a, P. Ribeiro, T. Prosen and D. Bernard, Symmetry Classes of Classical Stochastic Processes, J. Stat.  Phys. \textbf{192}, 41 (2025).

\bibitem{Schu94} G. Sch\"utz and S. Sandow,
Non-abelian symmetries of stochastic processes: derivation of correlation functions  for random vertex models and disordered interacting many-particle systems, Phys. Rev. E \textbf{49}, 2726--2744 (1994).


\bibitem{Schu99} 
G.M. Sch\"utz, 
Non-equilibrium relaxation law for entangled polymers, Europhys. Lett. {\bf 48}, 623--628 (1999).

\bibitem{Schu03} 
G.M. Sch\"utz,
Critical phenomena and universal dynamics in one-dimensional
driven diffusive systems with two species of particles,
J. Phys. A \textbf{36}, R339 - R379 (2003).

\bibitem{Schu17}
G.M. Sch\"utz and B. Wehefritz–Kaufmann,
Kardar-Parisi-Zhang modes in d-dimensional directed polymers,
 Physical Review E \textbf{96}, 032119 (2017).

\bibitem{Serre}
Serre D. Systems of Conservation Laws 2: Geometric Structures, Oscillations, and Initial-Boundary Value Problems. Cambridge University Press; 1999.

\bibitem{Serr00}
D. Serre, Systems of Conservation Laws. Cambridge University Press, Cambridge (2000)

\bibitem{Spit70}
F. Spitzer, 
Interaction of Markov processes,
Adv. Math. \textbf{5}, 246--290, (1970)

\bibitem{Spoh83}
H. Spohn, Long-range correlations for stochastic lattice gases in a non-equilibrium steady state. J. Phys. A: Math. Gen. \textbf{16}, 4275--4291 (1983).

%
\bibitem{Spoh26}
H. Spohn, 
The Popkov-Schütz two-lane lattice gas: universality for general jump rates,
J. Stat. Mech. \textbf{2026}, 023203 (2026).

\bibitem{Toth03}
B. T\'oth and B. Valk\'o,
Onsager relations and Eulerian hydrodynamic limit for systems with several conservation laws, 
J. Stat. Phys. \textbf{112}, 497--521 (2003).

\bibitem{Tsarev}
Tsarev SP. The geometry of Hamiltonian systems of hydrodynamic type. The generalized hodograph method. Mathematics of the USSR-Izvestiya. 1991 Apr 30;37(2):397-419.

\bibitem{Uchi08} 
Uchiyama, M. 
Two-species asymmetric simple exclusion process with open boundaries. 
Chaos Solitons Fractals 35, 398--407,  (2008).

\bibitem{Wei00} Wei, Q.-H., Bechinger, C. and Leiderer P.:  Single-File diffusion of colloids in one-dimensional channels. Science \textbf{287}, 625--627 (2000).



\bibitem{Yau91}
H.-T. Yau,
Relative entropy and hydrodynamics of Ginzburg--Landau models,
Lett. Math. Phys. \textbf{22}, 63--80 (1991).

\bibitem{zahra2025hydrodynamics}
A. Zahra, 
Hydrodynamics of multi-species driven diffusive systems with open boundaries: a two-tasep study,
\textit{arXiv preprint} arXiv:2512.10108, 2025.

\bibitem{ZahraMulti}
A. Zahra, 
Multi-species generalization of the totally asymmetric simple exclusion process. arXiv preprint arXiv:2301.04066. 2023.





\end{thebibliography}
\end{document}